\documentclass[sigconf,nonacm]{aamas}

\usepackage{balance} %
\usepackage{multirow}
\usepackage{tikz}
\usetikzlibrary{arrows.meta, decorations.pathmorphing,bending,calc, decorations.pathreplacing}
\usepackage{enumitem}
\usepackage{thmtools, thm-restate}
\usepackage{subcaption}
\usepackage[title]{appendix}
\usepackage{etoolbox}
\usepackage{xcolor}
\usepackage{algorithm}
\usepackage[noend]{algpseudocode}
\usepackage{pifont}
\usepackage{pict2e}
\usepackage{orcidlink}

\makeatletter
\gdef\@copyrightpermission{
  \begin{minipage}{0.2\columnwidth}
   \href{https://creativecommons.org/licenses/by/4.0/}{\includegraphics[width=0.90\textwidth]{by}}
  \end{minipage}\hfill
  \begin{minipage}{0.8\columnwidth}
   \href{https://creativecommons.org/licenses/by/4.0/}{This work is licensed under a Creative Commons Attribution International 4.0 License.}
  \end{minipage}
  \vspace{5pt}
}
\makeatother

\setcopyright{ifaamas}
\acmConference[AAMAS '25]{Proc.\@ of the 24th International Conference
on Autonomous Agents and Multiagent Systems (AAMAS 2025)}{May 19 -- 23, 2025}
{Detroit, Michigan, USA}{Y.~Vorobeychik, S.~Das, A.~Nowé  (eds.)}
\copyrightyear{2025}
\acmYear{2025}
\acmDOI{}
\acmPrice{}
\acmISBN{}

\acmSubmissionID{27}

\title{The Strong Core of Housing Markets with Partial Order Preferences}

\author{Ildik\'o Schlotter%
}
\affiliation{
  \institution{HUN-REN Centre for Economic and Regional Studies}
  \city{Budapest}
  \country{Hungary}
  }
\email{schlotter.ildiko@krtk.hun-ren.hu}

\author{Lydia Mirabel Mendoza-Cadena%
}
\affiliation{
  \institution{MTA-ELTE Matroid Optimization Research Group, Department of Operations Research, \\
  E\"otv\"os Lor\'and University}
  \city{Budapest}
  \country{Hungary}}
\email{lyd21@student.elte.hu}

\begin{abstract}
We study the strong core of housing markets when agents' preferences over houses are expressed as partial orders. We provide a structural characterization of the strong core, and propose an efficient algorithm that finds an allocation in the strong core or decides that it is empty, even in the presence of forced and forbidden arcs. 
The algorithm satisfies the property of group-strategyproofness. 
Additionally, we show that certain results known for the strong core in the case when agents' preferences are weak orders can be extended to the setting with partial order preferences; among others, we show that the strong core in such housing markets satisfies the property of respecting improvements. 
\end{abstract}


\newif\ifshort
\shortfalse
         
\colorlet{undomcolor}{black}
\definecolor{domcolor}{RGB}{68,118,170}%
\definecolor{Xcolor}{RGB}{34,136,51} %
\definecolor{Xprimecolor}{RGB}{170,51,119}%

\algrenewcommand\algorithmicrequire{\textbf{Input:}}
\algrenewcommand\algorithmicensure{\textbf{Output:}}

\newenvironment{varalgorithm}[1]
  {\algorithm\renewcommand{\thealgorithm}{#1}}
  {\endalgorithm}

\setlist{itemsep=0pt}

\def\StrCoreFA{\textsc{Strong Core with Forbidden Arcs}}
\def\StrCoreFFA{\textsc{Strong Core with Forbidden and Forced Arcs}}
\def\S{\mathcal{S}}
\def\T{\mathcal{T}}
\def\R{\mathcal{R}}
\def\SC{\mathcal{SC}}
\def\SCFA{\mathsf{SCFA}}
\def\NN{T^\star}
\def\allNN{\widetilde{T}}

\def\P{$\mathsf{P}$}
\def\NP{$\mathsf{NP}$}
\newcommand{\EE}{\makebox{$\raisebox{.05em}{\makebox[0pt][l]{%
   $\exists\hspace{-.517em}\exists\hspace{-.517em}\exists$}}%
   \exists\hspace{-.517em}\exists\hspace{-.517em}\exists\,$}}
\def\strongcheckmark{{\footnotesize\ding{52}$^{\phantom{i}\text{\bf \EE}}$}}
\def\checkmark{{\footnotesize\ding{52}}}
\def\xmark{{\footnotesize\ding{55}}}

\newcommand{\opentriangle}{%
  \raisebox{0.2pt}{\makebox[0.77778em]{%
    \setlength{\unitlength}{0.6em}%
    \linethickness{0.4pt}\roundjoin
    \begin{picture}(1,1)
    \polygon(0,0)(1,0)(1,1)
    \end{picture}%
  }}%
}

\newenvironment{proofsketch}{\par
  \pushQED{\hfill \opentriangle}%
  \normalfont \topsep 12pt\relax
  \trivlist
  \item[\indent\hskip\labelsep
        {\scshape Proof sketch.}]\ignorespaces
}{%
  \popQED\endtrivlist
}

\newenvironment{claimproof}{\par
  \pushQED{$\hfill \lhd$}%
  \normalfont \topsep 6pt\relax
  \trivlist
  \item[\indent\hskip\labelsep
        {\scshape Claim proof.}]\ignorespaces
}{%
  \popQED\endtrivlist
}

\makeatletter
\newcommand{\leqnomode}{\tagsleft@true\let\veqno\@@leqno}
\makeatother

\theoremstyle{acmplain}
\newtheorem{claim}{Claim}

\theoremstyle{acmdefinition}
\newtheorem{remark}{Remark}

\begin{document}

\pagestyle{fancy}
\fancyhead{}

\maketitle 

\section{Introduction}

A Shapley-Scarf housing market involves a set of agents, each endowed with exactly one unit of some indivisible good -- a \emph{house} -- and having preferences over the houses owned by other agents in the market. Agents are assumed to trade among themselves without monetary transfers; hence, the outcome in such a market is an \emph{allocation} that assigns to each agent exactly one house, possibly its own. Housing markets have been extensively studied in the fields of economics and computer science since the seminal work of Shapley and Scarf in 1974~\cite{shapley-scarf-1974}. Their most prominent motivation comes from kidney-exchange programs~\cite{roth-sonmez-unver-2005,biro-kidney-exchange-survey,biroetal2021}, but applications include various exchange markets\footnote{For an online exchange market for board games, see  \url{https://boardgamegeek.com/wiki/page/Math_Trades}.} 
ranging from time-sharing markets~\cite{WangKrishna} and time banks~\cite{ACEE21} to on-campus housing~\cite{adbulkadiroglu-sonmez}.

Two prominent solution concepts widely investigated in connection to housing markets are the \emph{core} and the \emph{strong core}:
the \emph{core} of a housing market contains allocations from which no coalition of agents can deviate so that each of them strictly improves their situation by trading among themselves, 
while
the \emph{strong core} (or \emph{strict core}) contains allocations that do not admit a coalition of agents who are able to weakly improve their situation, with at least one agent strictly improving as a result of the deviation. 
Shapley and Scarf described the Top Trading Cycle (TTC) mechanism, attributed to David Gale, that always finds an allocation in the core~\cite{shapley-scarf-1974}. 
Roth and Postlewaite showed that if agents' preferences are strict, then the TTC mechanism always returns the unique allocation in the strong core~\cite{roth-postlewaite}. 
Roth proved that in such a model the TTC is \emph{strategyproof}, meaning that no agent can improve its outcome under TTC by misreporting its preferences~\cite{Roth82}. 
Later, Bird~\cite{Bird84} showed that TTC under strict preferences is even \emph{group-strategyproof}, that is, no set of agents can improve their situation by misreporting their preferences in a coordinated fashion. 

If agents can be indifferent between houses, then both the core and the strict core loose many of their appealing properties.
Although the TTC mechanism always returns an allocation in the core even when agents' preferences over the  houses are \emph{weak orders} (i.e., linear orders containing ties)~\cite{shapley-scarf-1974},
the core in such markets may contain multiple allocations. 
As opposed to the core, the strong core can become empty if agents' preferences are weak orders. Quint and Wako characterized the strong core in such housing markets, and gave a polynomial-time algorithm that either finds an allocation in the strong core or concludes that the strong core is empty~\cite{WakoQuint}. 
Several researchers have introduced generalizations of the TTC in search of a mechanism that maintains the desirable properties of TTC even under the presence of indifference in the market; this line of research has yielded algorithms that efficiently find allocations in the core with additional properties such as Pareto-optimality, while ensuring also strategyproofness~\cite{Alcalde-Unzu-Mollis,Jaramillo-Manjunath,Aziz-deKeijzer,Plaxton-2013,Saban-Sethuraman-2013}.

Despite the significant effort to deal with weakly ordered preferences in housing markets, 
there is still %
 limited knowledge about the case when agents' preferences can be  \emph{partial orders}, a generalization of weak orders. Partial orders arise in contexts where agents may regard two alternatives as \emph{incomparable}; 
as opposed to indifference in weak orders, incomparability is not necessarily transitive. Despite the clear motivations for studying partial orders in the context of kidney exchange, as we explain it below, there is scant literature on housing markets that allow agents to express their preferences using partial orders.
A recent study~\cite{SBF2024core} addresses questions about the core in such a model, but we are not aware of any research concerning the strong core in a setting where 
preferences can be partial orders.
This paper aims 
to address this gap in the literature.

\smallskip
\noindent
{\bf Motivation.}
Partial orders naturally arise in real-life applications where agents compare alternatives based on multiple unrelated criteria. An important example of this occurs in kidney exchange, where patients exchange the incompatible organs donated by their own donors; such a situation can be modelled as a housing market where agents are incompatible patient--donor pairs. The preferences of patients over possible kidney transplants are primarily based on two factors: 
the age of the donor and the HLA compatibility between donor and recipient. In fact, in the UK kidney exchange program, these two factors are the only ones for which patients can set a threshold of acceptance~\cite{biroetal2021}. The unrelated nature of these two factors gives rise to a partial ordering: we can only assert that a certain graft is better than another if it is more advantageous regarding both of these factors,
so a graft from a younger donor but with a worst HLA-matching is incomparable with a graft from an older donor but with a better HLA-matching. Such incomparability is not necessarily transitive; hence, patients' preferences over possible transplants is only a partial, and not a weak, ordering.

Another source of partial order preferences in kidney exchange (or, in fact, in many other applications) is the fact that small differences between two alternatives may not be significant enough to influence the preferences, but slight differences may add up to a more substantial contrast. For example, when the age difference between two possible donors is, say, less than a year, then a recipient might regard the corresponding grafts as equivalent; however, such incomparability is not transitive:  donors $A$ and~$B$, as well as $B$ and~$C$, might have an age difference less than a year, but $A$ might be almost two years older than~$C$. Such preferences can be described as so-called \emph{semi-orders}, a special case of partial orders.

Another application where partial orders 
arise naturally is 
a project aimed at promoting social inclusion of individuals with intellectual disabilities, %
where participants are intellectually disabled people who have to be assigned to various job positions, namely, to computer-numerical-controlled (CNC) machines in a facility~\cite{wiesner2014preconditions}.\footnote{Strictly speaking, this application does not directly concern housing markets; however, assuming a dynamic setting where  certain participants are already assigned to job positions while new participants may also enter the market, we quickly arrive at a situation that can be modelled as a housing market.} The preferences of the participants over the available positions are determined by multiple criteria, based on their abilities and the requirements of the job, and thus induce a partial ordering~\cite{bruggemann2017matchingCopeland}.
See~\cite{fattore2017BookPartial} for further applications where partial orders appear.

\smallskip
\noindent
{\bf Our contribution and organization.} We consider the strong core of housing markets where agents' preferences are expressed as partial orders, which 
is a setting that
has not been studied before. 

After providing the preliminaries in Section~\ref{sec:prelim}, 
we explain why the Quint--Wako characterization fails when agents' preferences are partial---and not weak---orders (Section~\ref{sec:QW-fails}), and present an alternative characterization of the strong core in such a setting (Section~\ref{sec:mychar}). 
Even though our characterization does not immediately yield an algorithm for finding an allocation in the strong core (if exists), in Section~\ref{sec:finding-strongcore} we propose a polynomial-time algorithm for this problem. Recall that the Quint--Wako algorithm~\cite{WakoQuint} finds an allocation in the strong core when agents' preferences are weak orders, or detects that the strong core is empty; hence, our algorithm can be thought of as a generalization of this result to the setting with partial order preferences.
In fact, our algorithm can handle forbidden and forced arcs in the graph underlying the market: it can find an allocation in the strong core that contains all forced arcs but no forbidden arcs, if such an allocation exists; see Sections~\ref{sec:algo}--\ref{sec:extension-to-forced-arcs}.
We prove that our algorithm is group-strategyproof (Section~\ref{sec:strategyproof}), and we also examine 
its ability to enumerate all allocations in the strong core (Section~\ref{sec:algo-prop}). 

In Section~\ref{sec:properties}, we investigate certain properties of the strong core. %
The main result of Section~\ref{sec:properties} is that the strong core satisfies the property of \emph{respecting improvement}: assuming that the outcome of a market is an allocation in the strong core, we show that an agent can only (weakly) benefit from an increase in the ``value'' of its house, i.e., from its house becoming more preferred according to the remaining agents; this observation is a generalization of the analogous result for weakly ordered preferences by Bir\'o et al.~\cite{BKKV-mor}. 

Section~\ref{sec:conclusion} contains our conclusions.

We defer some of our results and proofs to
\ifshort the full version of our paper~\cite{schlotter-mendoza-cadena_arxiv_fullversion};
the symbol~$\star$ marks statements with deferred proofs.
\else
the appendix; statements with deferred proofs are marked with the symbol~$\star$.
\fi

\smallskip
\noindent
{\bf Related work.} 
We have already mentioned in the introduction the series of papers containing prior work on the core and the strong core of housing markets, a line of research into which our paper fits smoothly. The closest works to our study are the papers by Quint and Waco~\cite{WakoQuint} and by Schlotter et al.~\cite{SBF2024core}. 
See Table~\ref{tab:related_work} for a comparison of the results in these works.

\begin{table}[t]
\caption{Summary of results on the problem of finding allocations in the core and strong of housing markets.
The phrase ``f/f arcs'' refers to problem variants with forced or forbidden arcs. ``\NP-h'' stands for \NP-hardness. If the problem is polynomial-time solvable, we write either ``in $\mathsf{P}$'' or the name of the corresponding algorithm. 
}
\label{tab:related_work}
    \centering
    \begin{tabular}{@{}l@{\hspace{2pt}}@{\hspace{2pt}}c@{\hspace{2pt}}c@{\hspace{2pt}}c@{}}
        Problem& 
        Strict orders &
        Weak orders & 
        Partial orders \\ \toprule
        Core        
        & TTC \cite{shapley-scarf-1974} 
        & TTC \cite{shapley-scarf-1974}
        & TTC \cite{shapley-scarf-1974,SBF2024core} \\
        Core, f/f arcs
        &$\mathsf{NP}$-h~\cite{SBF2024core} 
        &$\mathsf{NP}$-h~\cite{SBF2024core} 
        &$\mathsf{NP}$-h~\cite{SBF2024core} \\
        Strong core
        & TTC~\cite{roth-postlewaite} 
        & Quint--Waco~\cite{WakoQuint}
        & {\bf Thm~\ref{thm:alg-SCFA-correct}} \\
        Strong core, f/f arcs 
        & TTC~\cite{roth-postlewaite} 
        & in \P~\cite{SBF2024core}   
        & {\bf Thm~\ref{thm:alg-SCFA-correct}}   
    \end{tabular}
\end{table}

Besides the literature on housing markets, our paper also builds on various ideas that have been explored in the broader context of matching under preferences.

    First, we emphasize that even though we are not aware of any previous work on housing markets with partial order preferences with the exception of the work by Schlotter et al. on the core~\cite{SBF2024core}, partial order preferences have been widely studied in the area of matching under preferences, appearing already in early works on stable matchings~\cite{Fleiner-Irving-Manlove,Manlove-DAM02,IMS-stacs2003,FleinerIM-TCS,RastegariCILB-EC13}.

    Second, problems involving edge restrictions have also been extensively studied for stable matching problems, see e.g.~\cite{Knuth1976,FleinerIM-TCS,Cseh-Manlove-SR-2016,Dias-2003,Cseh-Heeger-2020}. In the context of housing markets, the question whether a given agent can obtain a certain house (or can avoid ending up with a certain house) has been addressed in~\cite{SBF2024core}. The analogous problems have also been studied in the house allocation model where a set of objects needs to be allocated to a set of agents (without initial endowments) using serial dictatorship~\cite{SabanSethuraman-MOR15,AzizBB-WINE13,AzizMestre-14,AzizWX-ijcai15,CFS-ADT2017}.
    
    Third, questions about the property of respecting improvement have been already studied in the context of housing markets~\cite{roth-sonmez-unver-2005,SBF2024core}, but this line of research has its roots in a study by Balinski and S\"onmez~\cite{balinski-sonmez} on college admission. The authors investigated how an improvement in a student's test scores may have undesirable effects. Their results inspired further studies of the same flavor about school choice and other allocation problems~\cite{hatfield-kojima-narita,sonmez-switzer,klaus-klijn-RI}.
    
    Fourth, the topic of (group-)strategyproofness has a vast literature in the area of social choice, and has been explored in connection to the TTC mechanism and its generalizations~\cite{Roth82,Bird84,Ma94,Alcalde-Unzu-Mollis,Jaramillo-Manjunath,Plaxton-2013,Saban-Sethuraman-2013}; more recent results investigate strategyproofness in relation to different extensions of housing markets~\cite{HongPark-22,YDTLY-AAMAS22}, and the  interesting topic of \emph{obvious strategyproofness} in connection to TTC~\cite{Troyan19,MandalRoy22}.

\section{Preliminaries
\label{sec:prelim}}
For a positive integer~$\ell \in \mathbb{N}$, we use the notation $[\ell]=\{1,2,\dots,\ell\}$.

\smallskip 
\noindent
{\bf Directed graphs.}
We assume the reader is familiar with basic concepts from graph theory; see 
\ifshort the full version 
\else Appendix~\ref{app:prelim}
\fi 
for a detailed description of all the definitions we rely on.
Given a vertex $v$ in a directed graph~$D=(V,E)$, we let $\delta^{\text{in}}(v)=\{(u,v):(u,v) \in E\}$ and $\delta^{\text{out}}(v)=\{(v,u):(v,u) \in E\}$ denote the set of \emph{incoming} and \emph{outgoing} arcs of~$v$, respectively. 
Given a set~$F \subseteq E$ of arcs, we extend this notation by letting $\delta^{\text{in}}_F(v)=\delta^{\text{in}}(v) \cap F$
and $\delta^{\text{out}}_F(v)=\delta^{\text{out}}(v) \cap F$.
Given a vertex set $V' \subseteq V$ and an arc set $E' \subseteq E$, we let $D[V']$ denote the subgraph \emph{induced} by~$V'$ and $D[E']$ for the subgraph \emph{spanned} by~$E'$.
A notion of central importance for our purposes is the following: we call a set~$S \subseteq V$ an \emph{absorbing set} in~$D$, if $S$ is a strongly connected component of~$D$ %
that is not left by any arc.
In other words, when contracting~$S$, the newly introduced vertex corresponding to~$S$ has no outgoing arcs.

\smallskip 
\noindent
{\bf Partial orders.}
A partial ordering over a set~$S$ is a binary relation~$\succ$ over~$S$ that is irreflexive (i.e., $s \not\succ s$ for all $s \in S$), antisymmetric (i.e., if $s \succ r$, then $r \not\succ s$), and transitive (i.e., $s \succ r$ and $r \succ q$ implies $s \succ q$).
Given a partial ordering~$\succ$ over~$S$, two elements $s$ and~$r$ of~$S$ are  \emph{comparable} if $s \succ  r$ or $r \succ s$; otherwise, they are \emph{incomparable}, denoted by $s \sim r$. We write $s \succeq r$ for $r  \not\succ s$.

Partial orders are a generalization of weak orders. A partial ordering over~$S$ is a \emph{weak order} if and only if the incomparability relation~$\sim$ is an equivalence relation: it is reflexive, symmetric, %
and transitive. Note that the first two conditions always hold, hence the transitivity of the incomparability relation is the key property for a partial order to be a weak order. For example, $S = \{ s,r,q \}$ with $s \succ q$ but $s \sim r $ and $r \sim q$ is not a weak order, because $s \not\sim q$.

\smallskip 
\noindent
{\bf Housing markets.}
A \emph{housing market} is a pair $(N,\{\succ_a:a \in N\})$ where $N$ is the set of \emph{agents} %
and $\succ_a$ is a partial order over $N$ for each $a \in N$, representing the preferences of agent~$a$. 
Note that we identify each house with its initial owner and thus express agents' preferences over  houses as partial orders over~$N$.
Recall that 
we write $b \sim_a c$ if $b \not\succ_a c$ and  $c \not\succ_a b$ both hold, and $b \succeq_a c$ means that $c \not\succ_a b$.
We say that $a$ \emph{prefers}~$c$ to~$b$ if $c \succ_a b$, and $a$ \emph{weakly prefers}~$c$ to~$b$, if $c \succeq_a b$. Agent~$b$ is \emph{acceptable} to agent~$a$, if $b \succeq_a a$, and we let $A(a)=\{b \in N: b \succeq_a a\}$ denote the \emph{acceptability set} of~$a$. 

The underlying graph of a housing market~$H={(N,\{\succ_a:a \in N\})}$ is defined as the directed graph $D^H=(N,E)$ where vertices correspond to agents, and arcs correspond to acceptability, that is, $E=\{(a,b): a \in N, b \in A(a)\}$. Note that $(a,a) \in E$ for each~$a \in N$, so $D^H$ contains a loop at each agent.
We say that an arc~$(a,b) \in E$ is \emph{dominated} by an arc~$(a',b')$ if $a=a'$ and $b' \succ_a b$. An arc~$(a,b) \in E$ is \emph{undominated}, if no arc in~$E$ dominates it.

Given a subset~$N' \subseteq N$ of agents, the \emph{submarket} of~$H$ restricted to the agent set~$N'$ is the housing market~$H_{|N'}=(N',\{ \succ'_a: a \in N'\})$ where $\succ'_a$ is the restriction of~$\succ_a$ to the set~$N'$ of agents.

An \emph{allocation} in~$H$ is defined as a set~$X \subseteq E$ of arcs in~$D^H$ such that each agent has exactly one outgoing and exactly one incoming arc in~$X$, i.e., $|\delta^{\text{in}}_X(a)|=|\delta^{\text{out}}_X(a)|=1$ for each $a \in N$. 

\begin{remark}
\label{rem:find-alloc}
Given an arbitrary subset~$E'$ of arcs in~$D^H$, it is possible to check in polynomial time whether $H$ admits an allocation containing only arcs of~$E'$, by reducing this problem via a well-known, simple reduction to finding a perfect matching in a bipartite graph. 
This classic reduction constructs an auxiliary graph~$G=(N \cup N',F)$ where $N'=\{a':a \in N\}$ is a copy of the set~$N$, and the edge set is $F=\{ ab': (a,b)  \in E'\}$; then allocations in~$H$ contained in~$E'$ correspond bijectively to perfect matchings in~$G$.
\end{remark}

For an allocation~$X$, or more generally, a set~$X$ of arcs in~$D^H$ such that $|\delta^{\text{out}}_X(a)| \leq 1$ for each agent~$a \in N$, we let $X(a)$ denote the agent~$b$ if $\delta^{\text{out}}_X(a)=\{(a,b)\}$; if $\delta^{\text{out}}_X(a)=\emptyset$, then we set ${X(a)=\varnothing}$.
We extend the agents' preferences over~$N$ to preferences over allocations in the straightforward way: agent~$a$ (weakly) prefers an allocation~$X$ to an allocation~$X'$ 
if and only if $a$ (weakly) prefers $X(a)$ to~$X'(a)$. Extending the notation, we write $X \succ_a X'$ for $a$ preferring~$X$ to~$X'$; we use $X \succeq_a X'$ and $X \sim_a X'$ analogously. 

An allocation~$X$ is in the \emph{strong core} of~$H$, if there does not exist a \emph{blocking cycle} for~$X$, i.e., a cycle~$C$ in $D^H$ such that all agents~$a$ on~$C$ weakly prefer~$C(a)$ to~$X(a)$, and at least one agent~$a^\star$ on~$C$ (strictly) prefers~$C(a^\star)$ to~$X(a^\star)$. We remark that such cycles are often called \emph{weakly} blocking cycles. 
We adopt the notation by Bir\'o et al.~\cite{BKKV-mor} and write $\SC(H)$ for the strong core of~$H$.

An allocation~$X$ is in the \emph{core} of~$H$, if there does not exist a \emph{strictly blocking cycle} for~$X$, i.e., a cycle~$C$ in $D^H$ such that all agents~$a$ on~$C$ strictly prefer~$C(a)$ to~$X(a)$.

\section{Characterizing the strong core}
\label{sec:characterization}
If agents' preferences are strict orders, then the strong core always contains a unique allocation~\cite{roth-postlewaite}, which can be found in linear time using the Top Trading Cycle algorithm, attributed to David Gale in the seminal paper by Shapley and Scarf~\cite{shapley-scarf-1974}. If preferences are weak orders, then the strong core can be empty, but Quint and Wako provided a characterization of allocations in the strong core in such a setting, and showed how to use this characterization to find an allocation in the strong core if there exists one~\cite{WakoQuint}.

In Section~\ref{sec:QW-fails} we show that the characterization of the strong core by Quint and Wako fails to hold in the setting when agents' preferences can be partial orders. In Section~\ref{sec:mychar} we generalize the results of Quint and Wako to the setting with partial order preferences by giving a characterization of the strong core.

\subsection{The Quint--Wako Characterization Fails for Partial Orders}
\label{sec:QW-fails}
For housing markets where agents' preferences are weak orders, Quint and Wako proved the following characterization of allocations in the strong core~\cite{WakoQuint}. 
Recall that an absorbing set is a strongly connected component from which no arc may leave.

\begin{theorem}[{\bf Quint--Wako characterization~\cite{WakoQuint}}]
\label{thm:QW-char}
Suppose that $H=(N,\{\succ_a:a \in N\})$ is a housing market where each $\succ_a$ is a weak order. 
Let $U$ be the set of undominated arcs in~$D^H$, 
and $S$ an absorbing set in~$D^H[U]$. 
An allocation~$X$ for~$H$ is in the strong core of~$H$ if and only if it can be partitioned into sets~$X_S$ and~$X_{N \setminus S}$ where 
\begin{itemize}
\item[(a)] 
$X_S=X \cap (S \times S)$ is an allocation in $H_{|S}$ contained in~$U$, and
\item[(b)] 
$X_{N \setminus S}=X \cap (N \setminus S) \times (N \setminus S)$ is an allocation in the strong core of~$H_{|(N\setminus S)}$.
\end{itemize}
\end{theorem} 

It is not hard to show that conditions~(a) and~(b) stated in Theorem~\ref{thm:QW-char} are \emph{necessary} for a given allocation to be in the strong core also in the case when agents' preferences are partial orders (see 
\ifshort
    the full version 
\else
    Appendix~\ref{app:necessary}
\fi    
for details).
Example~\ref{ex1} shows that conditions~(a) and~(b) of Theorem~\ref{thm:QW-char} are \emph{not sufficient} to ensure that an allocation belongs to the strong core. 
Thus, the characterization established by Theorem~\ref{thm:QW-char} for weakly ordered preferences fails to hold in the case when agents' preferences are partial orders. 
\begin{example}
\label{ex1}
Consider the following market~$H^1$ over agent set $N=\{a,b,c,d\}$; see its underlying graph in Figure~\ref{fig:ex1}. Let the acceptability sets be defined as $A(a)=A(b)=N$, $A(c)=\{a,b,c\}$, and $A(d)=\{c,d\}$. 
For each agent~$x$ and $y \in A(x) \setminus \{x\}$, we let $y \succ_x x$; additionally, $c \succ_a d$ and $c \succ_b d$. 
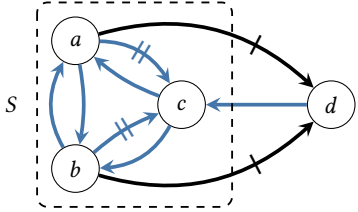
\begin{figure}[th]
\centering
\begin{tikzpicture}[scale=.85]
	\node[draw,circle,inner sep=1.5pt, minimum size=18pt] (a) at (0,1) {$a$};
  	\node[draw,circle,inner sep=1.5pt, minimum size=18pt] (b) at (0,-1) {$b$};
  	\node[draw,circle,inner sep=1.5pt, minimum size=18pt] (c) at (1.66,0) {$c$};
	\node[draw,circle,inner sep=1.5pt, minimum size=18pt] (d) at (4,0) {$d$};  	
  	\draw[line width=1.4pt,arrows={-stealth},color=domcolor] (a) to [bend left=10] (b);
  	\draw[line width=1.4pt,arrows={-stealth},color=domcolor] (b) to [bend left=10] node[midway,sloped] {$\boldsymbol{||}$} (c);
  	\draw[line width=1.4pt,arrows={-stealth},color=domcolor] (c) to [bend left=10] (a);
  	\draw[line width=1.4pt,arrows={-stealth},color=domcolor] (a) to [bend left=30] node[midway,sloped] {$\boldsymbol{||}$} (c);
  	\draw[line width=1.4pt,arrows={-stealth},color=domcolor] (b) to [bend left=30] (a);
  	\draw[line width=1.4pt,arrows={-stealth},color=domcolor] (c) to [bend left=30] (b);
  	\draw[line width=1.4pt,arrows={-stealth},color=undomcolor] (a) to [bend left=30] node[pos=0.7,sloped] {$\boldsymbol{|}$} (d);
  	\draw[line width=1.4pt,arrows={-stealth},color=undomcolor] (b) to [bend right=30] node[pos=0.7,sloped] {$\boldsymbol{|}$} (d);
  	\draw[line width=1.4pt,arrows={-stealth},color=domcolor] (d) to (c);
  	\draw[rectangle,line width=0.7pt,dashed,rounded corners] (-0.6,1.6) rectangle (2.45,-1.6) {};
  	\node[] at (-1,0) {$S$};
\end{tikzpicture}
\Description{The underlying graph of housing market~$H^1$ defined in Example~\ref{ex1}.}
\caption{The underlying graph of housing market~$H^1$ defined in Example~\ref{ex1}. Henceforth, loops are omitted, and undominated arcs are shown in blue, the strongly connected components they form are shown as dashed polygons. 
Single and double line markings $\boldsymbol{|}$ and $\boldsymbol{||}$ convey domination: an arc marked with~$\boldsymbol{||}$ dominates an arc marked with~$\boldsymbol{|}$ and leaving the same agent;
e.g., $(a,c)$ dominates $(a,d)$ but not $(a,b)$, and neither does $(a,b)$ dominate $(a,d)$. 
}
\label{fig:ex1}
\end{figure}

Note that $U=\{(a,b),(a,c),(b,c),(b,a),(c,a),(c,b),(d,c)\}$ is the set of undominated arcs, and the unique absorbing set in $D^{H^1}[U]$ is $S=\{a,b,c\}$. The submarket $H^1_{|S}$ admits two allocations in~$U$: 
$X_1=\{(a,b),(b,c),(c,a)\}$  and $X_2=\{(a,c),(c,b),(b,a) \}$. The only allocation for the submarket~$H^1_{|(N\setminus S)}=H_{|\{d\}}$ is ~$X_d=\{(d,d)\}$.

Although both $X_1$ and~$X_2$ are in the strong core of~$H^1_{|S}$, and $X_d$  is in the strong core of~$H_{|N \setminus S}$, neither $X_1 \cup X_d$ nor $X_2 \cup X_d$ is in the strong core of~$H^1$: 
$\{(a,d),(d,c),(c,a)\}$ is a blocking cycle for $X_1 \cup X_d$, and similarly, $\{(b,d),(d,c),(c,b)\}$ is a blocking cycle for $X_2 \cup X_d$.
This shows that Theorem~\ref{thm:QW-char} does not extend to partial order preferences.
\end{example}

\subsection{A Characterization Through Peak Sets}
\label{sec:mychar}
Given an allocation~$X$ in~$H$, let us define the following arc sets:
\begin{align*}
E_{[\sim X]} & = \{(a,b) \in E: X(a) \sim_a b\}, \\
E_{[\succ X]} & = \{(a,b) \in E: b \succ_a X(a)\}, \qquad \textrm{ and}\\
E_{[\succeq X]} & = E_{[\sim X]} \cup E_{[\succ X]}.
\end{align*}
Let $D_{[\sim X]}=(N,E_{[\sim X]})$, and let the digraphs $D_{[\succ X]}$ and $D_{[\succeq X]}$ be defined analogously.
The following definition captures the central notion we need for our characterization of the strong core.

\begin{definition} 
\label{def:peakset}
Given an allocation~$X$ in a housing market~$H$, 
a \emph{peak set} of~$X$ is an absorbing set in~$D_{[\succeq X]}$. 
\end{definition}

In our characterization of the strong core, peak sets will play the role that absorbing sets in the subgraph of undominated arcs have in Theorem~\ref{thm:QW-char}.
The following lemma states their key property.

\begin{restatable}[$\star$]{lemma}{lemstrcorepeakset}
\label{lem:strcore-peakset}
Suppose that $X$ is an allocation in the strong core of~$H$.
If $S$ a peak set of~$X$, then $X \cap (S \times S)$ is an allocation in~$H_{|S}$ 
consisting only of undominated arcs.
Moreover, 
each arc leaving~$S$ is dominated by an arc of~$X$.
\end{restatable}

Based on Lemma~\ref{lem:strcore-peakset}, we can now characterize the strong core.

\begin{theorem}
\label{thm:peakset-char}
Suppose that $H=(N,\{\succ_a:a \in N\})$ is a housing market where each~$\succ_a$ is a partial order, and let $U$ be the set of undominated arcs in~$D^H$. Then allocation~$X$ for~$H$ is in the strong core of~$H$ if and only if it can be partitioned into two sets~$X_1$ and~$X_2$ such that there exists a set~$S \subseteq N$ of agents for which
\begin{itemize}
\item[(a)] $X_1$ is an allocation in~$H_{|S}$ contained in~$U$,
\item[(b)] $X_2$ is an allocation in the strong core of $H_{|(N\setminus S)}$, and
\item[(c)] each arc of~$D^H$ leaving~$S$ is dominated by an arc of~$X_1$.
\end{itemize}
Moreover, $S$ can be chosen as a peak set for~$X$.
\end{theorem}

\begin{proof}
First suppose that $X$ is in the strong core of~$H$. Define~$S$ as a peak set of~$X$, and set $X_1:=X \cap (S \times S)$
and $X_2=X \setminus X_1$. By Lemma~\ref{lem:strcore-peakset} we know that conditions~(a) and~(c) are fulfilled. Since~$X_1$ is an allocation on~$H_{|S}$, we necessarily have that $X_2=X \setminus X_1$ is an allocation in~$H_{|(N \setminus S)}$. To see that $X_2$ is in the strong core of this submarket, it suffices to note that any blocking cycle for~$X_2$ in~$H_{|(N \setminus S)}$ would also block~$X$ in~$H$. 

Second, suppose that some allocation~$X$ can be partitioned into $X_1$ and $X_2$ satisfying conditions~(a)--(c) for some set~$S$ of agents; we show that $X$ is in the strong core of~$H$. Assume for the sake of contradiction that there is a blocking cycle~$C$ for~$X$. First, since $X_1$ is an allocation for~$H_{|S}$ and $X_1 \subseteq U$, we know that $C$ must involve at least one agent not in~$S$. Second, since $X_2$ is an allocation in the strong core of~$H_{|(N \setminus S)}$, we know that $C$ must also contain at least one agent not in~$N \setminus S$. Thus, $C$ involves agents from both~$S$ and $N \setminus S$, and consequently has to contain at least one arc leaving~$S$. However, all arcs leaving~$S$ are dominated by an arc of~$X_1$, which contradicts the assumption that $C$ blocks~$X$. This proves that no blocking cycle for~$X$ exists, so~$X$ is indeed in the strong core of~$H$.
\end{proof}

\section{Finding an allocation in the strong core}
\label{sec:finding-strongcore}
In this section, we show how Theorem~\ref{thm:peakset-char} can be used to construct an algorithm that either finds an allocation in the strong core of a housing market where agents' preferences are partial orders, or concludes that such an allocation does not exist. 

Before explaining our ideas, let us briefly describe the algorithm by Quint and Wako~\cite{WakoQuint} that solves this problem in the setting where agents' preferences are weak orders, based on  Theorem~\ref{thm:QW-char}.
Given a housing market $H=(N,\{\succ_a:a \in N\})$, 
the Quint--Wako algorithm first computes an absorbing set~$S$ in the subgraph~$D^H[U]$ of undominated arcs in the underlying digraph.
Then it checks whether $H_{|S}$ admits an allocation consisting only of undominated arcs, and if so, stores such an allocation~$X_S$ and proceeds with the submarket~$H_{|(N \setminus S)}$ of the remaining agents recursively.
If the recursive call returns an allocation~$X_{N \setminus S}$ in the strong core of~$H_{|(N \setminus S)}$, then the algorithm outputs $X_S \cup X_{N \setminus S}$; otherwise it concludes that the strong core of~$H$ is empty.

Unfortunately, the characterization of the strong core given in Theorem~\ref{thm:peakset-char} does not readily offer an algorithm for finding an allocation in the strong core of a given housing market, or to decide whether such an allocation exists. The reason for this is that it is not immediately clear how to find a set~$S$ of agents for which conditions~(a)--(c) of Theorem~\ref{thm:peakset-char} hold for the desired allocation~$X$ (or, more precisely, for the arcs sets~$X_1$ and $X_2$ that together yield~$X$). By Theorem~\ref{thm:peakset-char} we know that the set~$S$ can be chosen to be a peak set of the desired allocation; however, this insight does not seem to offer a direct way to find such a set without knowing~$X$ itself. 

In this section, we propose a method that resolves this issue and finds an allocation in the strong core of the given housing market, if such an allocation exists.
In fact, we are going to solve the following more general computational problem:

\ifshort
    \begin{center}
    \fbox{ 
    \parbox{0.94\columnwidth}{
    \StrCoreFA{} (SCFA): 
    
    \textbf{Input:} A housing market~$H$ with partial orders and a set~$F$ of \emph{forbidden} arcs in the underlying directed graph~$D^H$. 
    
    \textbf{Task:}  Decide if there exists an allocation~$X$ in the strong core of~$H$ such that $X \cap F=\emptyset$, and if so, find such an allocation.
    }}
    \end{center}
\else
    \begin{center}
    \fbox{ 
    \parbox{0.94\columnwidth}{
    \begin{tabular}{l}\StrCoreFA{} (SCFA):  \end{tabular} \\
    \begin{tabular}{p{2em}p{.78\columnwidth}}
    Input: & A housing market~$H$ with partial orders and a set~$F$ of \emph{forbidden} arcs in the underlying directed graph~$D^H$. \\
    Task: & Decide if there exists an allocation~$X$ in the strong core of~$H$ such that $X \cap F=\emptyset$, and if so, find such an allocation.
    \end{tabular}
    }}
    \end{center}
\fi

In Section~\ref{sec:algo} we describe a polynomial-time algorithm to solve the \StrCoreFA{} problem, and we sketch its proof of correctness in Section~\ref{sec:algo-correctness}. We extend our algorithm to a setting with forced arcs in Section~\ref{sec:extension-to-forced-arcs}, and explain its group-strategyproofness in Section~\ref{sec:strategyproof}.
In Section~\ref{sec:algo-prop} we briefly compare it with the Quint--Wako algorithm~\cite{WakoQuint}.

\subsection{Algorithm for SCFA}
\label{sec:algo}

We now present Algorithm~\ref{alg:SCFA} to solve our instance~$(H,F)$ of SCFA. An allocation in $\SC(H)$ that is disjoint from~$F$ is a \emph{solution}.

\begin{varalgorithm}{SCFA}
\caption{Solving SCFA.}
\label{alg:SCFA}
\begin{algorithmic}[1]
\Require{An instance~$(H,F)$ of SCFA where $D^H=(N,E)$.}
\Ensure{An allocation in the strong core of~$H$ disjoint from~$F$, or~$\varnothing$ if no such allocation exists.}
\If{$|N|=1$} 
	\If{$(a,a) \in E \setminus F$ where $N=\{a\}$} {\bf return} $\{(a,a)\}$.\label{line:one-agent}
	\Else { {\bf return} $\varnothing$.}
	\EndIf
\EndIf
\State Let $U$ denote the set of undominated arcs in~$D^H$.
\State Let $\mathcal{S}$ be the set of strongly connected components in $(N,U)$.\label{line:compute-SCCs}
\State Let $\T=\emptyset$.
\ForAll{$S \in \S$} \label{line:T-start}
	\State Let $E_{S}=\{(a,b) \in U \setminus F: a,b \in S\}$.\label{line:def-A_S}
	\If{$\exists$ an allocation for~$H_{|S}$ in $E_S$} Put $S$ into~$\T$. \label{line:add-S-to-T}
	\EndIf
\EndFor
\While{$|\T|>0$} \label{line:iter-begin}
	\State Let $\R=\emptyset$ and $T^\star=\bigcup_{S \in \T} S$.\label{line:def-R-Tstar}
	\ForAll{$S \in \T$}
		\State Let $E_{S,\T}=\!\{(a,b) \in E_S \colon b \!\succ_a \! b' 
    \,\, \forall
        (a,b') \in E,b' \notin T^\star\}$.\label{line:def-A_ST}
		\If{$\exists$ an allocation for~$H_{|S}$ in $E_{S,\T}$}
  \State Let $X_S$ be such an allocation. \label{line:Tvalid-def}
		\Else { Add $S$ as an element to~$\R$.}\label{line:add-S-to-R}
		\EndIf
	\EndFor
	\If{$\R=\emptyset$} 
         \If{$N=T^\star$} {{\bf return} $\bigcup_{S \in \T} X_S$. } \label{line:return-without-recursion}
         \EndIf
		\State Let $N'=N \setminus T^\star$ and $F'=F \cap (N' \times N'$).\label{line:remove-agents}
		\State Call $X'=\mathbf{SCFA}(H_{|N'},F')$. \label{line:recurse}
		\If{$X'=\varnothing$} {\bf return} $\varnothing$  \label{line:reject-recursively}
		\Else { {\bf return} $\bigcup_{S \in \T} X_S \cup X'$.} \label{line:outputX}
		\EndIf
	\Else { Delete each $S \in \R$ from~$\T$.} \label{line:iter-end}
	\EndIf
\EndWhile 
\State {\bf return} $\varnothing$. \label{line:reject}
\end{algorithmic}
\end{varalgorithm}

\smallskip
\noindent {\bf Description of Algorithm~\ref{alg:SCFA}.}
The main idea of %
our algorithm
is to find a set~$T^\star$ that can serve as the peak set of our desired allocation in the strong core, find an appropriate allocation on~$H_{|T^\star}$, and solve the remainder recursively after deleting all agents in~$T^\star$.

The algorithm starts by dealing with the case when the market only contains a single agent.
Next, it computes the family~$\S$ of strongly connected components in the digraph~${D^H[U]=(N,U)}$ of undominated arcs and,
using the method of
 Remark~\ref{rem:find-alloc},
 checks for each component~$S \in \S$ whether there is an allocation for the submarket $H_{|S}$ consisting only of undominated non-forbidden arcs. %
Those sets that pass this test are collected in a family~$\T \subseteq \S$. 

Given this collection~$\T \subseteq \S$, an iterative process is started. 
At each step in this iteration, Algorithm~\ref{alg:SCFA} checks whether each   submarket~$H_{|S}$ for $S \in \T$ admits an allocation
that contains only arcs that (i) are undominated and not forbidden, and (ii) dominate all arcs leaving~$T^\star=\bigcup_{T \in \T} T$; the set of such arcs in denoted by~$E_{S,\T}$ on line~\ref{line:def-A_ST}. More formally, let us say that an allocation~$X$ in~$H_{|S}$ is \emph{$T^\star$-valid} for some $T^\star \subseteq N$, if \begin{itemize}
\item[(i)]  $X \subseteq U \setminus F$
and 
\item[(ii)] $X(a) \succ_a b'$ for each $(a,b') \in E$ where $a \in T^\star$ and $b' \notin T^\star$. 
\end{itemize}
Thus, an allocation~$X$ in~$H_{|S}$ is $T^\star$-valid if and only if 
\[ X \subseteq \left\{(a,b) \in U \setminus F: b \succ_a b' \textrm{ for each } (a,b') \in E, b' \notin T^\star\right\}.\]
If Algorithm~\ref{alg:SCFA} finds that all submarkets~$H_{|S}$, $S \in \T$, admit a $T^\star=\bigcup_{T \in \T} T$-valid allocation, then it stores such an allocation for each~$S \in \T$, deletes the agents in~$T^\star$, and proceeds with the remaining market recursively---unless the remaining market is empty, in which case no recursion is necessary. If, on the contrary, $H_{|S}$ does not admit a $T^\star$-valid allocation for certain sets~$S \in \T$, then each such set~$S$ is removed from the collection~$\T$. 

This iterative process stops either if all sets in~$\T$ pass the test of admitting a $\bigcup_{T \in \T} T$-valid allocation, or if $\T$ becomes empty, in which case the algorithm concludes that there is no solution.%

\begin{example}
\label{ex2}
Consider again the housing market~$H^1$ defined in Example~\ref{ex1} and depicted in Figure~\ref{fig:ex1}, to see how Algorithm~\ref{alg:SCFA} runs with $H^1$ as its input. For simplicity, we define the set of forbidden arcs as empty. 
The set of strongly connected components in the subgraph spanned by all undominated arcs is $\S=\{S_1,S_2\}$ where $S_1=\{a,b,c\}$ and  $S_2=\{d\}$. 
 Algorithm~\ref{alg:SCFA} computes the arc sets~$E_{S_1}=\{(a,b),(b,c),(c,a),(a,c),(c,b),(b,a)\}$ and $E_{S_2}=\emptyset$, and finds that $H^1_{|S_1}$ admits an allocation in~$E_{S_1}$, e.g., the allocation $\{(a,b),(b,c),(c,a)\}$, but no allocation for~$H^1_{|S_2}$ (containing only the vertex~$d$) can be constructed from the empty set. Thus, 
Algorithm~\ref{alg:SCFA} initializes the family~$\T$ by setting $\T=\{S_1\}$.

Starting the iteration on lines~\ref{line:iter-begin}--\ref{line:iter-end}, the algorithm first computes the arc set~$E_{S_1,\T}$, obtaining  ${E_{S_1,\T}=\{(a,c),(c,a),(b,c),(c,b)\}}$; note that neither~$(a,b)$ nor~$(b,a)$ is contained in~$E_{S_1,\T}$, since they do not dominate the arc~$(a,d)$ or~$(b,d)$, respectively (observe that both~$(a,d)$ and~$(b,d)$ leave the set of vertices contained in~$\T$, that is, $T^\star=S_1=\{a,b,c\}$).
The algorithm next observes that $E_{S_1,\T}$ does not contain an allocation for~$H_{|S_1}$; thus, it removes~$S_1$ from~$\T$, leaving $\T$ empty. Hence, the iteration stops, and Algorithm~\ref{alg:SCFA} outputs~$\varnothing$ on line~\ref{line:reject}, concluding that the strong core is empty.
\end{example}
\begin{figure}[th]
\centering
\begin{tikzpicture}[scale=.85]
	\node[draw,circle,inner sep=1.5pt, minimum size=18pt] (a) at (0,1) {$a$};
  	\node[draw,circle,inner sep=1.5pt, minimum size=18pt] (b) at (0,-1) {$b$};
  	\node[draw,circle,inner sep=1.5pt, minimum size=18pt] (c) at (1.66,0) {$c$};
	\node[draw,circle,inner sep=1.5pt, minimum size=18pt] (d1) at (4,0.6) {$d_1$};  	
    \node[draw,circle,inner sep=1.5pt, minimum size=18pt] (d2) at (4,-0.6) {$d_2$};  	
  	\draw[line width=1.4pt,arrows={-stealth},color=domcolor] (a) to [bend left=10] (b);
  	\draw[line width=1.4pt,arrows={-stealth},color=domcolor] (b) to [bend left=10] node[midway,sloped] {$\boldsymbol{|}\boldsymbol{|}$} (c);
  	\draw[line width=1.4pt,arrows={-stealth},color=domcolor] (c) to [bend left=10] (a);
  	\draw[line width=1.4pt,arrows={-stealth},color=domcolor] (a) to [bend left=30] node[midway,sloped] {$\boldsymbol{|}\boldsymbol{|}$} (c);
  	\draw[line width=1.4pt,arrows={-stealth},color=domcolor] (b) to [bend left=30] (a);
  	\draw[line width=1.4pt,arrows={-stealth},color=domcolor] (c) to [bend left=30] (b);
  	\draw[line width=1.4pt,arrows={-stealth},color=undomcolor] (a) to [bend left=30] node[pos=0.7,sloped] {$\boldsymbol{|}$} (d1);
  	\draw[line width=1.4pt,arrows={-stealth},color=undomcolor] (b) to [bend right=30] node[pos=0.7,sloped] {$\boldsymbol{|}$} (d2);
  	\draw[line width=1.4pt,arrows={-stealth},color=domcolor] (d1) to (c);
  	\draw[line width=1.4pt,arrows={-stealth},color=domcolor] (d1) to [bend left=30] (d2);
  	\draw[line width=1.4pt,arrows={-stealth},color=domcolor] (d2) to [bend left=30] (d1);
    
  	\draw[rectangle,line width=0.7pt,dashed,rounded corners,above=2pt] (-0.6,1.6) rectangle (2.26,-1.7) {};
  	\node[] at (-0.9,0) {$S_1$};
   
    \draw[rectangle,line width=0.7pt,dashed,rounded corners,] (3.4,1.2) rectangle (4.6,-1.2) {};
   \node[] at (4.9,0) {$S_2$};

\end{tikzpicture}
\Description{The underlying graph of housing market~$H^2$ defined in Example~\ref{ex3}.}
\caption{The housing market~$H^2$ defined in Example~\ref{ex3}.  
}
\label{fig:ex3}
\end{figure}
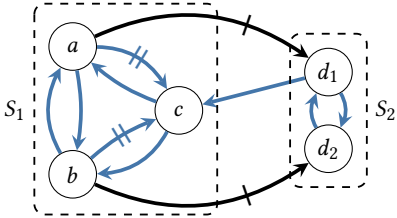

\begin{example}
\label{ex3}
Consider the following market~$H^2$ over agent set $N=\{a,b,c,d_1,d_2\}$; see its underlying graph (without loops) on Figure~\ref{fig:ex3}. Let the acceptability sets be defined as $A(a)=N \setminus\{d_2\}$, $A(b)=N \setminus\{d_1\}$, $A(c)=\{a,b,c\}$, $A(d_1)=\{c,d_1,d_2\}$, and $A(d_2)=\{d_1,d_2\}$. 
For each agent~$x \in N$ and $y \in A(x) \setminus \{x\}$, we let $y \succ_x x$; 
additionally, $c \succ_a d_1$ and $c \succ_b d_2$. 

Consider how Algorithm~\ref{alg:SCFA} runs with $H^2$ as its input, without forbidden arcs. 
All arcs are undominated except for~$(a,d_1)$ and $(b,d_2)$.
The strongly connected components in the subgraph spanned by all undominated arcs are 
$S_1=\{a,b,c\}$ and  $S_2=\{d_1,d_2\}$. 
 Algorithm~\ref{alg:SCFA} computes the arc sets $E_{S_1}=\{(a,b),(b,c),(c,a),\allowbreak(a,c),\allowbreak(c,b),(b,a)\}$ and $E_{S_2}=\{(d_1,d_2),(d_2,d_1)\}$, and finds that both~$H^2_{|S_1}$ and~$H^2_{|S_2}$ admit an allocation in~$E_{S_1}$ and in~$E_{S_2}$, respectively. Thus, 
Algorithm~\ref{alg:SCFA} initializes $\T$ by setting $\T=\{S_1,S_2\}$.

Starting the iteration on lines~\ref{line:iter-begin}--\ref{line:iter-end}, the algorithm first computes the arc sets~$E_{S,\T}$ for both $S \in \T$, obtaining  ${E_{S_1,\T}=E_{S_1}}$ and ${E_{S_2,\T}=E_{S_2}}$ and storing allocations~$X_{S_1}$ and~$X_{S_2}$ for~$H_{|S_1}$ and for~$H_{|S_2}$, respectively. Suppose that the algorithm finds $X_{S_1}\allowbreak=\{(a,b),\allowbreak(b,c),(c,a)\}$ and~$X_{S_2}=\{(d_1,d_2),(d_2,d_1)\}$; the case when $X_{S_1}=\allowbreak\{(a,c),\allowbreak(c,b),(b,a)\}$ is similar. The algorithm finds on line~\ref{line:return-without-recursion} that $N=T^\star=S_1 \cup S_2$, and hence returns the allocation $X_{S_1} \cup X_{S_2}$.
\end{example}

\subsection{Correctness of Algorithm~\ref{alg:SCFA}}
\label{sec:algo-correctness}
The correctness of Algorithm~\ref{alg:SCFA} relies on the following key fact.%
\begin{lemma}
\label{lem:peak-invariant}
If $X$ is an allocation in the strong core of~$H$ disjoint from~$F$, and $P$ is a peak set of~$X$, then $P \subseteq \bigcup_{T \in \T} T$ always holds on lines~\ref{line:iter-begin}--\ref{line:iter-end} during the execution of Algorithm~\ref{alg:SCFA} on~$(H,F)$.
\end{lemma}

\begin{proof}
First, recall that if $P$ is a peak set for~$X$, then no arc of~$E_{[\succeq X]}$ leaves~$P$ by definition. This implies that no undominated arc leaves~$P$, and therefore,
each strongly connected component of~$D^H[U]$ is either entirely contained in~$P$, or is disjoint from it. 
Hence, if $\S$ denotes the family of strongly connected components in~$D^H[U]$, then $P=\bigcup_{S \in \S_P} S$ for some $\S_P \subseteq \S$. 

By Lemma~\ref{lem:strcore-peakset}, the arc~$(p,X(p))$ is undominated for each $p \in P$. Since no undominated arc leaves~$P$, this implies that $X$ contains a collection of pairwise vertex-disjoint cycles in~$U$ covering all agents in~$P$. By the definition of strongly connected components, each such cycle must be entirely contained in some strongly connected component of~$D^H[U]$. Summarizing this, we obtain that for each $S \in \S_P$ the arc set~$X_S=\{(s,X(s)):s \in S\}$ is an allocation in~$H_{|S}$ that is contained in~$U$ and disjoint from~$F$.
Consequently, each set~$S \in \S_P$ is added to~$\T$ on line~\ref{line:add-S-to-T} of Algorithm~\ref{alg:SCFA}.
Hence, at the beginning of the iteration on lines~\ref{line:iter-begin}--\ref{line:iter-end}, $P \subseteq \bigcup_{T \in \T} T$ holds. 

We prove that this remains true throughout the run of the algorithm. For this, it suffices to see that whenever some set~$S \in \S$ is removed from~$T$, then $S \notin \S_P$. 
So assume that $S \in \S_P$ is removed from~$\T$ on line~\ref{line:iter-end}, and before the removal, $P \subseteq T^\star$ for $T^\star=\bigcup_{T \in \T}T$. 
Then $H_{|S}$ does not admit an allocation in~$E_{S,\T}$, that is, $H_{|S}$ does not admit a $ T^\star$-valid allocation. 

Recall now that allocation~$X_S$ for~$H_{|S}$ is $P$-valid due to Lemma~\ref{lem:strcore-peakset}. However, since $P \subseteq T^\star$, all arcs leaving~$T^\star$ necessarily leave~$P$ as well, and hence must be dominated by an arc of~$X_S$. It follows that $X_S$ is $T^\star$-valid. This contradicts our assumption that $H_{|S}$ does not admit a $T^\star$-valid allocation, proving that no set in~$\S_P$ is ever removed by the algorithm from~$\T$.
\end{proof}

The correctness of Algorithm~\ref{alg:SCFA} is stated below.

\begin{restatable}[$\star$]{theorem}{thmalgSCFAcorrect}
\label{thm:alg-SCFA-correct}
Algorithm~\ref{alg:SCFA} correctly solves each instance of SCFA, and runs in~$O(|N|^2 \cdot |E|^{1+o(1)})$ time where $(N,E)$ is the underlying graph of the housing market~$H$ of the instance.
\end{restatable}
\begin{proofsketch}
The proof applies induction on the number of agents; the correctness is clear for $|N|=1$.

We need to show that whenever Algorithm~\ref{alg:SCFA} concludes 
that the instance does not admit a solution, then this indeed holds.
First, if a solution exists, then the algorithm never reaches line~\ref{line:reject}  as the peak set of the solution remains in~$\bigcup_{T \in \T}T$ by Lemma~\ref{lem:peak-invariant}. Second, Algorithm~\ref{alg:SCFA} cannot return~$\varnothing$ on line~\ref{line:reject-recursively} either if a solution exists, due to the induction hypothesis.

For the other direction, we need to prove that whenever the algorithm returns a set~$X$ of arcs, then $X \in \SC(H)$. This can be shown using Theorem~\ref{thm:peakset-char} in combination with our induction hypothesis.

By $|\T| \leq |N|$, lines~\ref{line:iter-begin}--\ref{line:iter-end} are performed at most~$|N|$ times.
The bottleneck in each iteration is to compute a perfect matching in a bipartite graph, which takes $O(|E|^{1+o(1)})$ time~\cite{ChenKLPPS-CACM23-linear-maxflow}. 
The recursion yields an additional factor of~$|N|$ in the running time. 
\ifshort 
\else For the details of the proof, see 
 Appendix~\ref{app:proof-of-correctness}.
\fi
\end{proofsketch}

\begin{corollary}
\label{cor:SCFA}
SCFA with $n$ agents and $m$ arcs in the underlying graph can be solved in $O(n^2 \cdot m^{1+o(1)})$ time.
\end{corollary}

\subsection{Extension to Forbidden and Forced Arcs}
\label{sec:extension-to-forced-arcs}
We remark that Algorithm~\ref{alg:SCFA} can also be used to solve the extension of SCFA with \emph{forced} arcs. 

\ifshort
    \begin{center}
    \fbox{ 
    \parbox{0.94\columnwidth}{
   \StrCoreFFA{} (SCFFA): 
    
    \textbf{Input:} A housing market~$H$ with partial orders, a set~$F^+$ of \emph{forced} arcs and a set~$F^-$ of \emph{forbidden} arcs in the underlying directed graph~$D^H$. 
    
    \textbf{Task:}  Decide if there exists an allocation~$X \in \SC(H)$ %
    such that $F^+ \subseteq X$ and $X \cap F^-=\emptyset$, and if so, find such an allocation.
    }}
    \end{center}
\else
    \begin{center}
    \fbox{ 
    \parbox{0.94\columnwidth}{
    \begin{tabular}{l}\StrCoreFFA{} (SCFFA):  \end{tabular} \\
    \begin{tabular}{p{2em}p{.78\columnwidth}}
    Input: & A housing market~$H$ with partial orders, a set~$F^+$ of \emph{forced} arcs and a set~$F^-$ of \emph{forbidden} arcs in the underlying directed graph~$D^H$. \\
    Task: & Decide if there exists an allocation~$X$ in the strong core of~$H$ such that $F^+ \subseteq X$ and $X \cap F^-=\emptyset$, and if so, find such an allocation.
    \end{tabular}
    }}
    \end{center}
\fi

To solve SCFFA in a housing market~$H$, it is sufficient to solve the variant that only involves forbidden arcs.
Indeed, since an allocation contains \emph{exactly one} outgoing arc for each agent, we know that an allocation in~$H$ contains an arc~$(a,b)$ of the underlying graph~$D^H$ if and only if it does not contain any of the arcs~$\delta^{\text{out}}(a) \setminus \{(a,b)\}$ in~$D^H$.
Therefore, we can reduce 
an instance of~$(H,F^+,F^-)$  of SCFFA
to an equivalent instance $(H,F)$ of SCFA
where 
$F=F^- \cup \bigcup_{(a,b) \in F^+}   \left( \delta^{\text{out}}_A(a) \setminus \{(a,b)\} \right). 
$
Thus, we obtain the following consequence of Corollary~\ref{cor:SCFA}.

\begin{corollary}
\label{cor:SCFFA}
SCFFA with $n$ agents and $m$ arcs in the underlying graph can be solved in $O(n^2 \cdot m^{1+o(1)})$ time.
\end{corollary}

\subsection{Strategyproofness}
\label{sec:strategyproof}
A desirable property of Algorithm~\ref{alg:SCFA} besides its efficiency is that it is
\emph{group-strategyproof}, i.e., it is not possible for a coalition of agents to improve their situation via misrepresenting their preferences. 

To formalize this definition, let $f_{\SCFA}(H,F)$ denote the set of allocations that can be obtained via Algorithm~\ref{alg:SCFA} on input~$(H,F)$; note that there can be several such allocations due to line~\ref{line:Tvalid-def}, where the algorithm may pick an arbitrary one among a set of possible allocations for a given submarket~$H_{|S}$. 
Given a housing market ${H=(N,\{\succ_a:a \in N\})}$ and a coalition~$C \subseteq N$ of agents, 
a \emph{$C$-deviation} from~$H$  is a housing market~$H'=(N,\{\succ'_a:a \in N\})$ where $\succ_a=\succ'_a$ for all $a \in N \setminus C$. 
\begin{definition}
    \label{def:groupSP}
    Let $f$ be a mechanism that assigns to each instance~$(H,F)$ of the \StrCoreFA{} problem a set~$f(H,F) \subseteq \SC(H)$ of solutions for $(H,F)$. We say that $f$ is \emph{group-strategyproof} (with respect to  SCFA) if there does not exist an instance~$(H,F)$ with $H=(N,\{\succ_a:a \in N\})$, a coalition~$C \subseteq N$ and a $C$-deviation~$H'$ from~$H$ such that 
    there are allocations $X \in f(H,F)$ and $X' \in f(H',F)$ for which $X'(c) \succ_c X(c)$ for each $c \in C$.
\end{definition}

We remark that Definition~\ref{def:groupSP} takes into account models deviations where agents may only misreport their preferences but have no possibility to change the set of forbidden arcs.

\begin{restatable}[$\star$]{theorem}{thmgroupSP}
    \label{thm:groupSP}
    Algorithm~\ref{alg:SCFA} is group-strategyproof.
\end{restatable}

\begin{proofsketch}
The proof relies on the observation that, roughly speaking, whenever an agent~$p$ is removed from the market during Algorithm~\ref{alg:SCFA}, the allocation fixed for~$p$ is the best $p$ can hope for, as long as no agent removed earlier than (or together with)~$p$ is contained in the deviating coalition~$C$.
Using this claim for the first agent~$c$ of~$C$ %
removed from the market, one can show that there is at least one agent removed before~$c$ for which the allocation obtained as a result of the deviation is disadvantageous; this gives rise to a blocking cycle, leading to a contradiction that proves the result. 
\ifshort 
\else See Appendix~\ref{app:proof-groupSP} for the details.
\fi 
\end{proofsketch}

\subsection{Finding All Allocations in the Strong Core}
\label{sec:algo-prop}

It is known that the Quint--Wako algorithm can find all allocations in the strong core of a housing market~$H=(N,\{ \succ_a\}_{a \in N})$ in which agents' preferences are described by weak orders. More precisely, if $S$ is an absorbing set of the graph~$D^H$ underlying~$H$, then by choosing an allocation (consisting of undominated arcs) for the submarket~$H_{|S}$ appropriately, and using also appropriate choices within the recursive call on~$H_{|(N \setminus S)}$, the Quint--Wako algorithm can return any fixed allocation~$X$ in the strong core of~$H$.

Interestingly, this also holds for Algorithm~\ref{alg:SCFA} if agents' preferences are weak orders; hence, Algorithm~\ref{alg:SCFA} is as powerful as the Quint--Waco algorithm when used for enumerating the strong core. By contrast, if agents' preferences are partial orders, then 
the strong core may contain allocations that can never be returned by Algorithm~\ref{alg:SCFA}. 
For details and examples, see 
\ifshort
the full version.
\else
Appendix~\ref{app:comparison_QW}
\fi

\section{Properties of the strong core
\label{sec:properties}}
In previous sections, we already mentioned that the strong core may be empty when we assume that agents' preferences are partial orders. If the strong core is not empty, then we are interested in knowing some of its most relevant characteristics.

\smallskip
\noindent
{\bf Incomparability of strong core allocations.} We start with the remarkable fact that all agents find any two allocations~$X$ and~$Y$ in the strong core incomparable. This is stated in Lemma~\ref{lem:strongcore-solutions-incomp}
whose proof relies on structural observations about the arc set obtained, roughly speaking, by letting each agent~$a$ choose between~$(a,X(a))$ and~$(a,Y(a))$ according to its 
\ifshort preferences.
\else preferences (see Appendix~\ref{app:proof-of-lemSCincomp}).
\fi
 
\begin{restatable}[$\star$]{lemma}{lemstrongcoresolutionsincomp}
    \label{lem:strongcore-solutions-incomp}
    Given a housing market $H=(N,\{\succ_a:a \in N\})$ with partial order preferences, and allocations $X$ and~$Y$ in the strong core of~$H$,
    we have $X \sim_a Y$ for each agent~$a \in N$.
\end{restatable}

An important consequence of Lemma~\ref{lem:strongcore-solutions-incomp} is that even though there might be allocations in the strong core that can never be found by Algorithm~\ref{alg:SCFA}, this is not really relevant from the viewpoint of the agents, as all agents find %
any two allocations in the strong core incomparable.
Lemma~\ref{lem:strongcore-solutions-incomp} also plays a crucial role in proving that the strong core respects improvement.

\smallskip\noindent
{\bf The effect of improvements on the strong core.}
We next aim to examine the effect of a change in the market when the house of some agent~$p$ improves---is such a change necessarily (weakly) advantageous for~$p$ in terms of the strong core? Such questions are essential to understand the incentives of agents to improve their endowments. In the context of a kidney exchange program, 
does bringing a ``better'' (e.g., younger) donor always benefits the patient, assuming that the %
program offers a strong core allocation?

To formalize the concept of improvement, consider two housing markets $H=(N,\{ \succ_a : a \in N \})$ and $H'=(N,\{  \succ'_a : a \in N \})$ over the same agent set~$N$. For two agents $p,q \in N$, we say that $ H'$ is a \emph{$(p,q)$-improvement} of~$H$ if $H'$ is obtained from~$H$ by shifting the position of~$p$ ``upward'' in the preferences of agent~$q$; see also Example~\ref{ex5}. Formally, we require that
\begin{enumerate}
    \item for each $a \in N \setminus \{q\}$, $\succ_a =  \succ'_a$, that is, we only allow changes in the preferences of agent~$q$;
    \item for each $a \in N \setminus \{p \}$, $p \succ_q a$ implies $p \succ'_q a$, 
    and $p \succeq_q a$ implies $p \succeq'_q a$, 
    that is, agent~$p$ may only become more preferred for agent~$q$ in $\succ'$, and not less preferred;
    \item for each $a,b \in N \setminus \{p \}$, we have $a \succ_q b$ if and only if $a \succ'_q b$, that is, the preferences for agent~$q$ remain unchanged when comparing two agents not involving~$p$.
\end{enumerate}

Now, a \emph{$p$-improvement} of~$H$ is the result of a sequence of $(p,q_i)$-improvements from $H$ for some series of agents $q_1,\dots,q_k \in N$.

\begin{figure}[ht]
\centering
    \begin{subfigure}{0.32\columnwidth}
    \centering
    \resizebox{1\columnwidth}{!}{       
        \begin{tikzpicture}[scale=1]
        	        	\node[draw,circle,inner sep=1.5pt, minimum size=18pt] (a) at (0,1) {$a$};
          	\node[draw,circle,inner sep=1.5pt, minimum size=18pt] (b) at (0,-1) {$b$};
          	\node[draw,circle,inner sep=1.5pt, minimum size=18pt] (c) at (2.5,1) {$c$};
        	\node[draw,circle,inner sep=1.5pt, minimum size=18pt] (d) at (2.5,-1) {$d$};  	
          	\draw[line width=1.4pt,arrows={-stealth},color=domcolor] (a) to [bend left=20] (b);
          	\draw[line width=1.4pt,arrows={-stealth},color=domcolor] (b) to [bend left=20] (a);
         	\draw[line width=1.4pt,arrows={-stealth},color=domcolor] (d) to [bend left=20] (c);
          	\draw[line width=1.4pt,arrows={-stealth},color=undomcolor] (c) to [bend left=20] node[midway,sloped] {$\boldsymbol{|}$} (d);
          	\draw[line width=1.4pt,arrows={-stealth},color=domcolor] (c) to [bend right=10] node[midway,sloped] {$\boldsymbol{||}$} (a);
            
          	\draw[rectangle,line width=0.7pt,dashed,rounded corners] (-0.6,1.5) rectangle (0.6,-1.5) {};

          	\draw[rectangle,line width=0.7pt,dashed,rounded corners] (1.9,1.5) rectangle (3.1,0.3) {};

          	\draw[rectangle,line width=0.7pt,dashed,rounded corners] (1.9,-1.5) rectangle (3.1,-0.45) {};

        \end{tikzpicture}
        }
        \caption{Original housing market~$H^3$.}
        \label{fig:ex5-orig}
    \end{subfigure}
    \hfill
    \begin{subfigure}{0.32\columnwidth}
    \centering
    \resizebox{1\columnwidth}{!}{       
        \begin{tikzpicture}[scale=1]
        	\node[draw,circle,inner sep=1.5pt, minimum size=18pt] (a) at (0,1) {$a$};
          	\node[draw,circle,inner sep=1.5pt, minimum size=18pt] (b) at (0,-1) {$b$};
          	\node[draw,circle,inner sep=1.5pt, minimum size=18pt] (c) at (2.5,1) {$c$};
        	\node[draw,circle,inner sep=1.5pt, minimum size=18pt] (d) at (2.5,-1) {$d$};  	
          	\draw[line width=1.4pt,arrows={-stealth},color=domcolor] (a) to [bend left=20] (b);
          	\draw[line width=1.4pt,arrows={-stealth},color=domcolor] (b) to [bend left=20] (a);
         	\draw[line width=1.4pt,arrows={-stealth},color=domcolor] (d) to [bend left=20] (c);
          	\draw[line width=1.4pt,arrows={-stealth},color=undomcolor] (c) to [bend left=20] node[midway,sloped] {$\boldsymbol{|}$} (d);
          	\draw[line width=1.4pt,arrows={-stealth},color=domcolor] (c) to [bend right=10] node[midway,sloped] {$\boldsymbol{||}$} (a);
          	\draw[line width=1.4pt,arrows={-stealth},color=domcolor] (b) to (c);

            \node (a') at (-0.5,1.38) {};
			\node (b') at (-0.5,-1.38) {};
			\node (b'') at (0.3,-1.38) {};
			\node (c'') at (3.1,0.6) {};
			\node (c') at (3.1,1.38) {};
			\draw[rounded corners=2mm,line width=0.7pt,dashed] (a'.north) -- (b'.south) -- (b''.south)  -- (c''.north) -- (c'.north) -- cycle;

          	\draw[rectangle,line width=0.7pt,dashed,rounded corners] (1.9,-1.5) rectangle (3.1,-0.45) {};        
        \end{tikzpicture}
        }
        \caption{Market~$H^4$: a $(c,b)$-improvement of~$H^3$.}
        \label{fig:ex5-improves}
    \end{subfigure}
    \hfill
    \begin{subfigure}{0.32\columnwidth}
    \centering
    \resizebox{1\columnwidth}{!}{       
        \begin{tikzpicture}[scale=1]
        	        	\node[draw,circle,inner sep=1.5pt, minimum size=18pt] (a) at (0,1) {$a$};
          	\node[draw,circle,inner sep=1.5pt, minimum size=18pt] (b) at (0,-1) {$b$};
          	\node[draw,circle,inner sep=1.5pt, minimum size=18pt] (c) at (2.5,1) {$c$};
        	\node[draw,circle,inner sep=1.5pt, minimum size=18pt] (d) at (2.5,-1) {$d$};  	
          	\draw[line width=1.4pt,arrows={-stealth},color=domcolor] (a) to [bend left=20] (b);
          	\draw[line width=1.4pt,arrows={-stealth},color=domcolor] (b) to [bend left=20] (a);
         	\draw[line width=1.4pt,arrows={-stealth},color=domcolor] (d) to [bend left=20] (c);
          	\draw[line width=1.4pt,arrows={-stealth},color=undomcolor] (c) to [bend left=20] node[midway,sloped] {$\boldsymbol{|}$} (d);
          	\draw[line width=1.4pt,arrows={-stealth},color=domcolor] (c) to [bend right=10] node[midway,sloped] {$\boldsymbol{||}$} (a);

          	\draw[line width=1.4pt,arrows={-stealth},color=domcolor] (a) to [bend right=20] (c);

			\node (a') at (-0.5,1.38) {};
			\node (b') at (-0.5,-1.38) {};
			\node (b'') at (0.3,-1.38) {};
			\node (c'') at (3.1,0.6) {};
			\node (c') at (3.1,1.38) {};

			\draw[rounded corners=2mm,line width=0.7pt,dashed] (a'.north) -- (b'.south) -- (b''.south)  -- (c''.north) -- (c'.north) -- cycle;

          	\draw[rectangle,line width=0.7pt,dashed,rounded corners] (1.9,-1.5) rectangle (3.1,-0.45) {};
        \end{tikzpicture}
        }
        \caption{Market~$H^5$, a $(c,a)$-improvement of~$H^3$.}
        \label{fig:ex5-empty}
    \end{subfigure}
    
    \Description{The underlying graphs of housing markets $H^3,H^4$ and $H^5$ defined in Example~\ref{ex5}.}
    \caption{Illustration for Example~\ref{ex5}. 
    }
\end{figure}
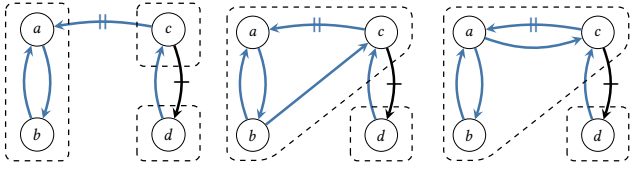

\begin{example}
\label{ex5}
Consider the following market~$H^3$ over agent set $N=\{a,b,c,d\}$; see Figure~\ref{fig:ex5-orig}. The acceptability sets are $A(a)=A(b)=\{a,b\}$, $A(c)=\{a,c,d\}$, and $A(d)=\{c,d\}$.
For each agent~$x$ and $y \in A(x) \setminus \{x\}$, we let $y \succ_x x$; 
additionally, we set $a \succ_c d$. 
Then $\SC(H^3)=\{X_{\textup{orig}}\}$ where $X_{\textup{orig}}=\{(a,b),(b,a),(c,c),(d,d)\}$. 

Let $H^4$ be the $c$-improvement  of~$H^3$ where the preferences are defined the same way as in~$H^3$ but $b$ now prefers~$c$ to its own house, so its acceptability set becomes $A(b)=N \setminus \{d\}$; see Figure~\ref{fig:ex5-improves}. Then %
$\SC(H^4)$ contains the unique allocation $\{(a,b),(b,c),(c,a),(d,d)\}$ which $c$ prefers to~$X_{\textup{orig}}$. 

Let now $H^5$ be the $c$-improvement  of~$H^3$ where the preferences are defined the same way as in~$H^3$ but $a$ now prefers~$c$ to its own house; see Figure~\ref{fig:ex5-empty}. The strong core of~$H^5$ then becomes empty, since the strongly connected component $\{a,b,c\}$ within the subgraph of undominated arcs does not contain an allocation. 
\end{example}

To study the effects of change in a housing market, we investigate how the strong core changes when a $p$-improvement occurs for some agent~$p \in N$. %
Particularly, we aim to understand whether a $p$-improvement is always beneficial (or, at least, not detrimental) to agent~$p$; hence, we aim to check the following property. 
\begin{definition}
\label{def:RI-prop-strongcore}
    The strong core of housing markets \emph{respects improvement}, if 
for all housing markets~$H$ and~$H'$ such that $H'$ is a $p$-improvement of~$H$ for some agent~$p$, we have $X' \succeq_p X$ for all allocations $X \in \SC(H,p)$ and $X' \in \SC(H',p)$.
\end{definition}
Note %
that 
Definition~\ref{def:RI-prop-strongcore} allows the strong core of~$H'$ to become empty as the result of an improvement, which may indeed happen if agents' preferences are weak or partial orders,
as we have already seen in Example~\ref{ex5}.

\begin{remark}
Notably, Lemma~\ref{lem:strongcore-solutions-incomp} does not hold for the core of a housing market with indifferences.
Hence, Bir\'o et al.~\cite{BKKV-mor} introduced the following notion: the core of housing markets \emph{respects improvement for the best available house}
 if 
for all housing markets~$H$ and~$H'$ such that $H'$ is a $p$-improvement of~$H$ for some agent~$p$, $X' \succeq_p X$ holds  whenever $X$ and $X'$ are among the most-preferred allocations by~$p$ within the core of~$H$ and of~$H'$, respectively;
the respecting improvement property for the \emph{worst} available house is defined analogously. These two notions coincide for the strong core with the property expressed in Definition~\ref{def:RI-prop-strongcore}.
Schlotter et al.~\cite{SBF2024core} proved that the core of housing markets respects improvement for the best available house, but 
not for the worst available house.
\end{remark}

Table~\ref{tab:summary} summarizes the results regarding the property of respecting improvements for the strong core and the core,  including Theorem~\ref{thm:RI-holds}, the main result of this section. Although the proof of Theorem~\ref{thm:RI-holds} is elementary in its technique, it requires certain insights and uses ideas from the proof of 
\ifshort Lemma~\ref{lem:strongcore-solutions-incomp}.
\else Lemma~\ref{lem:strongcore-solutions-incomp} (see Appendix~\ref{app:proof-of-RI} for details.
\fi 

\begin{restatable}[$\star$]{theorem}{thmRIholds}
    \label{thm:RI-holds}
    Let $H=(N,\{\succ_a:a \in N\})$ be a housing market in  which agents' preferences are partial orders, and let $H'$ be a $p$-improvement of~$H$ for some agent~$p$. 
    Then for all allocations~$X$ and~$X'$ in the strong core of~$H$ and~$H'$, respectively, 
    it holds that $X' \succeq_p X$. 
\end{restatable}

\begin{table}[th]   
\caption{Summary of results on the effect of improvements. RI stands for ``respects improvement'', while RI-best / RI-worst stands for ``respects improvement for the best / worst available house'', respectively. We use \checkmark\ to indicate that the given property  holds; \strongcheckmark\ additionally signals that the strong core (or core) cannot become empty as the result of a $p$-improvement; \xmark\  means that the given property fails to hold.\label{tab:summary}
}
    \centering
    \begin{tabular}{@{}l|c@{\hspace{6pt}}c@{\hspace{6pt}}c@{}}
Preferences & Strong core RI 
&  Core RI-best & Core RI-worst \\ 
\midrule
Strict orders	&	
\strongcheckmark~\cite{BKKV-mor}	&	
\strongcheckmark~\cite{SBF2024core} &	
\xmark~\cite{SBF2024core} \\
Weak orders	&	
\checkmark~\cite{BKKV-mor}	&	
\strongcheckmark~\cite{SBF2024core} &	
\xmark~\cite{SBF2024core} \\
Partial orders	&	
\checkmark~Thm.~\ref{thm:RI-holds} \,\,\,	&	
\strongcheckmark~\cite{SBF2024core} &	
\xmark~\cite{SBF2024core} 
    \end{tabular}
\end{table}

\balance

\smallskip\noindent
{\bf Formulation as an integer linear program.}
Quint and Wako were the first to provide an ILP for the strong core of housing markets~\cite{WakoQuint}. 
Subsequently, Bir\'o, Klijn, Klimentova, and Viana~\cite{BKKV-mor} provided an improved ILP formulation for the strong core
where the number of constraints is only linear in the size of the market. 
Bir\'o et al.~\cite{BKKV-mor} assumed weakly ordered preferences, hence it is not immediately clear whether their formulation remains correct for partial order preferences. 
We prove that the ILP formulation for the strong core by Bir\'o et al.\ remains sound even when agents' preferences are partial orders;
see 
\ifshort the full version
\else Appendix~\ref{app:ILP}
\fi for the details.

\section{Conclusion}
\label{sec:conclusion}
We investigated the strong core of housing markets in a model where agents' preferences are  partial orders, filling a gap in the literature. Our main result is a polynomial-time algorithm for obtaining an allocation in the strong core whenever the strong core is not empty; 
our algorithm is group-strategyproof and can accommodate forbidden and forced arcs.
We also discovered important properties of the strong core in housing markets with partial order preferences. 
In particular, we showed that the strong core in such housing markets respects improvements, and that 
the ILP proposed by Bir\'o et al.~\cite{BKKV-mor} for the strong core remains sound even when agents' preferences are partial orders.

As the strong core can be empty in housing markets with indifferences, it would be interesting to study allocations that are ``nearly in the strong core''. E.g., can we find allocations for which a few agents cover all blocking cycles? Such an allocation could be stabilized %
by compensating these agents, thus preventing deviations. Or can we find an allocation from which no deviation yields a strict improvement for at least $k$ agents for some constant $k \in \mathbb{N}$?

\begin{acks}
I. Schlotter is supported by the Hungarian Academy of Sciences under its Momentum Programme (LP2021-2) and its J\'anos Bolyai Research Scholarship.
L.M. Mendoza-Cadena was supported by the Lend\"ulet Programme of the Hungarian Academy of Sciences -- grant number LP2021-1/2021; and by the Ministry of Innovation and Technology of Hungary -- grant number ELTE TKP 2021-NKTA-62.
\end{acks}

\bibliographystyle{ACM-Reference-Format} 
\bibliography{core}

\ifshort

\else
\clearpage
\begin{appendices}
\section{Basics on graphs}
\label{app:prelim}
Consider a \emph{directed graph} $D=(V,E)$ where $V$ is the set of \emph{vertices} and $E \subseteq V \times V$ is the set of \emph{arcs}. 
An arc~$(u,v)$ leading from~$u$ to~$v$ has \emph{tail}~$u$ and \emph{head}~$v$. A \emph{loop} is an arc~$(v,v)$ for some~$v \in V$.
Given a set~$U \subseteq V$ of vertices, we say that an arc~$(u,v)$ \emph{leaves}~$U$ if $u \in U$ but $v \notin U$; 
similarly, $(u,v)$ \emph{enters}~$U$ if~$u \notin U$ but $v \in U$.

A set~$P \subseteq E$ is a \emph{path} in~$D$, if there exist distinct vertices $v_1,\dots,v_{\ell}$ in~$V$ such that $P=\{(v_i,v_{i+1}) :i \in [\ell-1]\}$; the vertices $v_1,\dots,v_\ell$ \emph{appear} on~$P$, and~$P$ \emph{starts} at~$v_1$ and \emph{ends} at~$v_\ell$, or equivalently, \emph{leads} from~$v_1$ to~$v_\ell$. An $(u,v)$-path is a path that starts at $u$ and ends at $v$.
A \emph{cycle} in~$D$ is the union of an arc~$(u,v)$ and a path leading from~$v$ to~$u$.
A graph~$D'=(V',E')$ is a \emph{subgraph} of~$D$ if $V'\subseteq V$ and $E' \subseteq E$; sometimes we may treat paths or cycles as subgraphs instead of arc sets.
The subgraph of~$D$ \emph{induced by} a set~$V' \subseteq V$ of vertices is obtained by deleting all vertices of~$V \setminus V'$ from~$D$ together with all incident arcs. 
The subgraph of~$D$ \emph{spanned by} a set~$E' \subseteq E$ of arcs is the graph~$(V,E')$.

A \emph{strongly connected component} of~$D$ is an inclusion-wise maximal set~$U \subseteq V$ of vertices such that for all $u,u' \in U$ there exists a path leading from~$u$ to~$u'$, and also a path leading from~$u'$ to~$u$.
By \emph{contracting} $U \subseteq V$, we mean the operation of adding a newly introduced vertex~$u^\star$, 
replacing the head of each arc entering~$U$ with~$u^\star$, 
replacing the tail of each arc leaving~$U$ with~$u^\star$, 
and then deleting the vertices of~$U$; note that this operation does not create loops. It is known that the strongly connected components of~$D$ yield a partition of~$V$, and contracting each strongly connected component results in an \emph{acyclic} digraph, i.e., a directed graph without cycles.
\section{Additional material for Section~\ref{sec:characterization}}
\subsection{Necessary Conditions for Strong Core Allocations}
\label{app:necessary}
Here we prove that conditions~(a) and~(b) stated in Theorem~\ref{thm:QW-char} are \emph{necessary} for a given allocation to be in the strong core also in the case when agents' preferences are partial orders, as stated in the following proposition.

\begin{proposition}
\label{prop:QW-necessary}
Suppose that $H=(N,\{\succ_a:a \in N\})$ is a housing market where each $\succ_a$ is a partial order. Again, let $U$ denote the set of undominated arcs in~$D^H$, and let $S$ be an absorbing set in~$D^H[U]$. 
If $X$ is an allocation in the strong core of~$H$, then  conditions~(a) and~(b) of Theorem~\ref{thm:QW-char} hold for~$X$.
\end{proposition}

\begin{proof}
Define $X_S=X \cap (S \times S)$ and $X_{N \setminus S}=X \setminus (S \times S)$, as in conditions~(a) and~(b). 

Suppose for the sake of contradiction that condition~(a) fails. This means that there exists some~$s \in S$ such that either the arc~$(s,X(s))$ leaves~$S$, or it is dominated by some arc of~$D^H$. Recall that $S$ is an absorbing set in~$D^H[U]$, and thus no arc of~$U$ leaves~$S$. 
Thus, in either of these cases we know $(s,X(s)) \notin U$, and hence there exists an arc~$e=(s,a) \in \delta^{\text{out}}_U(s)$ that dominates~$(s,X(s))$, i.e., $a \succ_s X(s)$. Since $e \in U$, we know that $a \in S$. Since~$S$ is a strongly connected component in~$D^H[U]$, there exists a cycle~$C$ within~$S$ through~$e$. Since~$s$ prefers~$C$ to~$X$, and by $C \subseteq U$ no agent on~$C$ prefers~$X$ to~$C$, we get that $C$ is a blocking cycle for~$X$; a contradiction. Hence, we obtain that $X_S$ is indeed an allocation in~$H_{|S}$ contained in~$U$. 

To see that condition~(b) holds as well, first note that by condition~(a), $X \setminus X_S=X_{N \setminus S}$, and so $X_S$ and $X_{N \setminus S}$ indeed partition $X$ into two. Moreover, $X_{N \setminus S}$ is an allocation in the submarket~$H_{|(N \setminus S)}$. 
Observe now that a blocking cycle for~$X_{}$ in this submarket would be a blocking cycle for~$X$ in~$H$ as well, proving that $X_{N \setminus S}$ is indeed in the strong core of~$H_{|(N \setminus S)}$.
\end{proof}

\subsection{Proof of Lemma~\ref{lem:strcore-peakset}}
\lemstrcorepeakset*
\begin{proof}
Let $U$ denote the set of undominated arcs in~$D$.

We first show that if $S$ is an absorbing set in~$D_{[\succeq X]}$ and $s \in S$, then the arc~$(s,X(s))$ is contained in~$U$. Assume for the sake of contradiction that $(s,X(s))$ is dominated by some arc of~$D^H$. Then $(s,X(s))$ is also dominated by some arc~$(s,s') \in U$. Now, since $(s,s') \in U$, we also know that $(s,s')$ is in $E_{[\succeq X]}$ and thus cannot leave~$S$, by the definition of an absorbing set, so we get $s' \in S$.
Since $S$ is a strongly connected component of~$D_{[\succeq X]}$, there must exists a path~$Q$ from~$s'$ to~$s$ within~$D_{[\succeq X]}[S]$. Now, since $(s,s')$ dominates $(s,X(s))$, the cycle obtained by adding $(s,s')$ to~$Q$ is a blocking cycle for~$X$, a contradiction. Hence, $(s,X(s)) \in U$ for each $s \in S$.

Since $U \subseteq E_{[\succeq X]}$, and no arc of~$E_{[\succeq X]}$ leaves~$S$ by the definition of an absorbing set, we obtain that $X(s) \in S$. Hence, $X \cap (S \times S)$ is indeed an allocation in~$H_{|S}$ contained in~$U$.

The last claim of the lemma follows from the fact that no arc of~$E_{[\succeq X]}$ may leave $S$ due to the definition of an absorbing set.
\end{proof}
\section{Additional material for Section~\ref{sec:finding-strongcore}}

\subsection{Proof of Theorem~\ref{thm:alg-SCFA-correct}}
\label{app:proof-of-correctness}
\thmalgSCFAcorrect*

\begin{proof}
Let $I=(H,F)$ denote our input instance.
We use induction on~$|N|$, the number of agent in~$I$, to prove the theorem.
Notice that if $|N|=1$, then the algorithm is clearly correct, since any solution must consist of a single loop that is adjacent to the unique agent in~$N$ and is not a forbidden arc---exactly the property checked on line~\ref{line:one-agent}.
Henceforth, let us assume that the algorithm correctly solves all instances with less than~$|N|$ agents. 

\begin{claim}
\label{clm:no-wrong-rejection}
If $I$ admits a solution, then Algorithm~\ref{alg:SCFA} does not reject~$I$ by outputting~$\varnothing$.
\end{claim}
\begin{claimproof}
Let $Y$ be a solution for~$I$, and let $P$ be a peak set for~$Y$. By Lemma~\ref{lem:peak-invariant}, $P \subseteq \bigcup_{T \in \T}T$ holds throughout lines~\ref{line:iter-begin}--\ref{line:iter-end} of Algorithm~\ref{alg:SCFA} when executed on instance~$I$; 
therefore, the algorithm never reaches line~\ref{line:reject}, and hence cannot output~$\varnothing$ on line~\ref{line:reject}.

Assume now for the sake of contradiction, that the algorithm outputs~$\varnothing$ on line~\ref{line:reject-recursively}, which happens only if the recursive call on line~\ref{line:recurse} returns~$\varnothing$. Observe that since $|\T|>0$ for the collection~$\T$ at this point, we know $|N'|<|N|$, and thus our inductive hypothesis implies that the instance $I'=(H_{|N'},F')$ does not admit a solution. 

Consider now the collection~$\T$ and the agent set~$T^\star=\bigcup_{T \in \T}T$. Observe that since each $S \in \T$ admits a $T^\star$-valid allocation~$X_S$, no undominated arc leaves~$T^\star$. 

We claim that no arc of~$Y$ leaves~$T^\star$ either. To see this, assume that some arc $(a,b) \in Y$ leaves~$T^\star$; then $(a,b)$ is dominated by the arc~$f=(a,X_S(a)) \in X_S$. Since $X_S$ contains only arcs within $H_{|S}$ and $S$ is a strongly connected component of~$D^H[U]$, we know that there is a path~$Q \subseteq U$ from the head of~$f$ to the tail of~$f$. The cycle~$C$ formed by $Q$ and the arc~$f$ is therefore contained in $E_{[\succeq Y]}$, with agent~$a$ strictly preferring~$C$ to~$Y$, i.e., $C$ is a blocking cycle for~$Y$, a contradiction. This proves that no arc of~$Y$ leaves~$T^\star$. However, then the restriction of~$Y$ to the set $N'=N \setminus T^\star$ of remaining agents, i.e., the arc set~$Y'=Y \cap (N' \times N')$ is an allocation in $H_{|N'}$. It is also clear that $Y'$ is in the strong core of~$H_{|N'}$, as a blocking cycle in~$H_{|N'}$ for~$Y'$ would also block~$Y$ in~$H$. Since $Y' \subseteq Y$ is disjoint from~$F$, we obtain that $Y'$ is a solution for the instance~$I'$, a contradiction proving the claim.
\end{claimproof}

Next, let us show the counterpart of Claim~\ref{clm:no-wrong-rejection}.
\begin{claim}
\label{clm:only-correct-outputs}
If Algorithm~\ref{alg:SCFA} outputs a set~$X \subseteq E$, then $X$ is solution for~$I$.
\end{claim}

\begin{claimproof}
Suppose that Algorithm~\ref{alg:SCFA} outputs a set~$X$ on line~\ref{line:return-without-recursion}. In this case, by $N=\bigcup_{S \in \T} S$ we know that $X$ is an allocation for~$H$. 
Moreover, we also have $X \subseteq \bigcup_{S \in \T} E_{S,\T} \subseteq U \setminus F$. Hence, it is clear that no blocking cycle for~$X$ exists, and therefore, $X$ is a solution for the input instance. 

Suppose now that Algorithm~\ref{alg:SCFA} outputs a set~$X$  on line~\ref{line:outputX}, and let us consider the values of~$\T$, $N'$ and $F'$ at this point. Due to our induction hypothesis, we know that $X'$ is a solution for the instance~$I'=(H_{|N'},F')$. By line~\ref{line:def-A_ST}, we also know that the allocation~$X_S$ for~$H_{|S}$ 
determined  on line~\ref{line:Tvalid-def}  
is $T^\star$-valid for each $S \in \T$ where $T^\star=\bigcup_{T \in \T}T$. 

First, it is clear that the set~$X=X' \cup \bigcup_{S \in \T} X_S$ is an allocation in~$H$: $X'$ is an allocation in~$H_{N'}$, and $\bigcup_{S \in \T} X_S$ is an allocation in~$H_{|T^\star}$. It is also obvious that $X$ does not contain any forbidden arcs, due to our induction hypothesis and the definition of $T^\star$-validity. It remains to show that $X$ is in the strong core of~$H$.

To show $X \in \SC(H)$, we apply Theorem~\ref{thm:peakset-char} with $X_1:=\bigcup_{S \in \T} X_S$ and $X_2:=X'$. 
First, since $X_1$ is an allocation in $H_{|T^\star}$ that only contains undominated arcs, it is immediate that $X_1 \in \SC(H_{|T^\star})$. 
Second, recall that $X_2 \in \SC(H_{N'})=\SC(H_{|N \setminus T^\star})$ by induction. 
Third, since the allocation~$X_S$ is $T^\star$-valid for each $S \in \T$, we get that each arc leaving $T^\star$ is dominated by an arc of~$X_1$. Hence, conditions~(a)--(c) of Theorem~\ref{thm:peakset-char} are satisfied,
proving that $X$ is indeed in the strong core of~$H$, and hence, it is a solution for~$I$.
\end{claimproof}

Due to Claims~\ref{clm:no-wrong-rejection} and~\ref{clm:only-correct-outputs}, to show the correctness of Algorithm~\ref{alg:SCFA} it suffices to prove that the algorithm always produces an output in the claimed running time. 

To see this, observe that $|\T|\leq |N|$, and thus lines~\ref{line:iter-begin}--\ref{line:iter-end} are performed at most~$|N|$ times, since Algorithm~\ref{alg:SCFA} either produces an output or decreases~$|\T|$ during each iteration. The bottleneck in each iteration---ignoring the time spent in the recursive call on line~\ref{line:recurse}---is to compute a perfect matching in a bipartite graph on~$2|N|$ vertices and at most~$|E|$ edges, which takes $O(|E|^{1+o(1)})$ time~\cite{ChenKLPPS-CACM23-linear-maxflow} 
 using the recent algorithm for maximum flow by Chen et al.~\cite{ChenKLPPS-CACM23-linear-maxflow}.
Therefore, the algorithm's running time without the recursive call on line~\ref{line:recurse} is $O(|N| \cdot |E|^{1+o(1)})$. 
The recursion introduces an additional factor of~$|N|$, because the number of agents decreases in each recursive call. This implies a total running time of~$O(|N|^2 \cdot |E|^{1+o(1)})$.
\end{proof}

\subsection{Proof of Theorem~\ref{thm:groupSP}}
\label{app:proof-groupSP}
\thmgroupSP*

\begin{proof}
Suppose for the sake of contradiction that there is an instance~$(H,F)$ of SCFA with $H=(N,\{\succ_a:a \in N\})$, a coalition~${C \subseteq N}$ of agents, a $C$-deviation $H'=(N,\{ \succ'_a:a \in N\})$ from~$H$, together with allocations $X \in f_\SCFA(H,F)$ and $X' \in f_\SCFA(H',F)$ for which $X'(c) \succ_c X(c)$ for each agent~$c \in C$ in the coalition.

Recall that Algorithm~\ref{alg:SCFA} run on input~$(H,F)$ repeatedly removes agents from~$H$, and performs a recursive call on the remaining submarket on line~\ref{line:recurse}.
Consider the run of the algorithm that produces~$X$ as its output.
Let $I_i=(H_{|N_i},F_i)$ be the input of the $i$-th call of the algorithm during this run, with $I_1$ being the original input instance, and each further call $I_i$, $i \geq 2$, initiated at line~\ref{line:recurse} of the previous call.
Let  $\NN_i$ denote the set of agents contained in~$T^\star$  (defined in line~\ref{line:def-R-Tstar}) at the point when the algorithm reaches line~\ref{line:return-without-recursion} in the $i$-th call. 
Then either the algorithm stops in the $i$-th call, or deletes the agents of~$\NN_i$ from the market on line~\ref{line:remove-agents}, right before the recursive call on~$I_{i+1}$, in which case we have
$N_{i+1}=N_i \setminus \NN_i$.

Consider the smallest~$k$ for which $\NN_k \cap C \neq \emptyset$, and let $c \in C \cap \NN_k$ be an agent in coalition~$C$ removed from the market in the $k$-th call. Note that the arc~$(c,X(c))$, which the algorithm will put into the returned allocation~$X$,
is undominated within the submarket~$H_{|N_k}$, due to the definitions of the arc sets~$E_S$ and $E_{S,\T}$ on lines~\ref{line:def-A_S} and~\ref{line:def-A_ST}.
By $X'(c) \succ_c X(c)$, we get that $X'(c)$ must have been removed earlier from the market, so $X'(c)=t$ for some agent~$t \in \allNN=\NN_1 \cup \dots \cup \NN_{k-1}$. 
By the definition of~$k$, we know $C \cap \allNN=\emptyset$. 
We prove the following.

\begin{figure}[ht!]
    \centering
    \begin{subfigure}{.5\textwidth}
    \centering
        \begin{tikzpicture}[
            myEdge/.style={draw,line width=1pt, -{Stealth[length=4mm]},  decorate,decoration={snake,post length=3mm,amplitude=1.8pt}},
            myArc/.style = { draw,line width = 1.5pt, -{Stealth[length=4mm]}}]
            \node[draw,circle,inner sep=1.5pt,minimum size=20pt] (a) at (3.6,2.8) {$a$};
        	\node[draw,circle,inner sep=1.5pt,minimum size=20pt] (a2) at (1.6,2.8) {$a'$};
        	\node[draw,circle,inner sep=4pt] (a3) at (3.6,1) {$ $};
          	\node[draw,circle,inner sep=1.5pt,minimum size=20pt] (b) at (0,2) {$b$};
          	\node[draw,circle,inner sep=4pt] (c) at (0,-0.2) {$ $};
          	\node[draw,circle,inner sep=4pt] (d) at (-2,2) {$ $};
          	\draw[myArc,Xprimecolor] (a) to node[pos=0.3] {$\boldsymbol{||}$}node[pos=0.35,above={2pt}] {$X'$} (a2);
            \draw[myArc,Xcolor] (a) to node[pos=0.3,sloped] {$\boldsymbol{|}$} node[pos=0.7, right] {$X$} (a3);
            \draw[myEdge,Xprimecolor,bend left=25] (a2) to node[pos=0.4,below right={-1pt}] {$X'$} (b);
            \draw[myArc,Xprimecolor] (b) to node[pos=0.3,sloped] {$\boldsymbol{|}$} node[pos=0.45,right] {$X'$} (c);
            \draw[myArc,Xcolor] (b) to %
            node[pos=0.3,sloped] {$\boldsymbol{||}$} (d);
            \draw[myEdge,Xcolor, out=100, in=90] (d) to node[below={3pt}] {$\widetilde{C} \subseteq X$} (b);
          	\draw[rectangle,line width=0.7pt,dashed,rounded corners,above=2pt] ($(d) - (0.5,1.5)$) rectangle ($(b) + (1.0,1.8)$) {};
          	\node (Th) at ($(d)!.5!(b) + (0.6,2.1)$) {$\NN_h$};
            \node at ($(Th) + (2.3,-0.03)$) {$\NN_{h-1} \; \cdots $};
            \draw[rectangle,line width=0.7pt,dashed,rounded corners,above=2pt] ($(a2) + (0.9,1)$) rectangle ($(a) + (0.7,-2.3)$) {};
            \node at ($(a) + (-0.2,1.25)$) {$\NN_{j}$};
        \end{tikzpicture}
        \caption{Illustration for Claim~\ref{claim:strategy-proofness}.}
        \label{fig:claim:strategy-proofness}
    \end{subfigure}
    \begin{subfigure}{.47\textwidth}
    \centering
        \begin{tikzpicture}[
            myEdge/.style={draw,line width=1pt, -{Stealth[length=4mm]},  decorate,decoration={snake,post length=3mm,amplitude=1.8pt}},
            myArc/.style = { draw,line width = 1.5pt, -{Stealth[length=4mm]}}]
            \node[draw,circle,inner sep=1.5pt,minimum size=20pt] (c) at (3.8,3.2) {$c$};
            \node[draw,circle,inner sep=4pt] (c2) at (3.8,1.4) {$ $};
            \node[draw,circle,inner sep=1.5pt,minimum size=20pt] (t) at (1.5,3.2) {$t$};
            \node[draw,circle,inner sep=1.5pt,minimum size=20pt] (tprime) at (0.5,2) {$t'$};
            \node[draw,circle,inner sep=4pt] (d) at (-1.5,2) {$ $};
            \node[draw,circle,inner sep=4pt] (f) at (0.5,0.2) {};
            \draw[myEdge,Xprimecolor,bend left=30] (t) to node[pos=0.3,below right] {$X'$}(tprime);
            \draw[myArc,Xprimecolor] (c) to node[pos=0.5,sloped] {$\boldsymbol{||}$}node[pos=0.65,above={3pt}] {$X'$} (t);
            \draw[myArc,Xcolor] (c) to node[pos = 0.55, right] {$X$} node[pos=0.3,sloped] {$\boldsymbol{|}$} (c2);
            \draw[myArc,Xprimecolor] (tprime) to node[pos=0.3,sloped] {$\boldsymbol{|}$}node[pos=0.35,right={3pt}] {$X'$} (f);
            \draw[myArc,Xcolor] (tprime) to node[pos=0.3,sloped] {$\boldsymbol{||}$}
            (d);
            \draw[myEdge,Xcolor, out=100, in=90] (d) to node[below={3pt}] {$\widehat{C} \subseteq X$} (tprime);
            \draw[rectangle,line width=0.7pt,dashed,rounded corners] ($(d) - (0.5,1.1)$) rectangle ($(t) + (1.4,1)$) {};
            \node (Ttilde) at ($(tprime) + (0.2,2.6)$) {$\widetilde{T}$};
            \draw[rectangle,line width=0.7pt,dashed,rounded corners,above=2pt] ($(t) + (1.8,1)$) rectangle ($(c) + (0.6,-2.3)$) {};
            \node at ($(c) + (0.2,1.3)$) {$\NN_{k}$};
        \end{tikzpicture}
        \caption{Illustration for Theorem~\ref{thm:groupSP}.}
        \label{fig:groupSP}
    \end{subfigure}
    \Description{Illustrations for Theorem~\ref{thm:groupSP}.}
    \caption{Illustrations for Theorem~\ref{thm:groupSP}. Paths are shown as zigzag lines. Arcs and paths belonging to allocation~$X$ (or~$X'$) are shown in \textcolor{Xcolor}{\bf green} (in \textcolor{Xprimecolor}{\bf purple}, respectively).}
    \label{fig:figures_group_SP}
\end{figure}
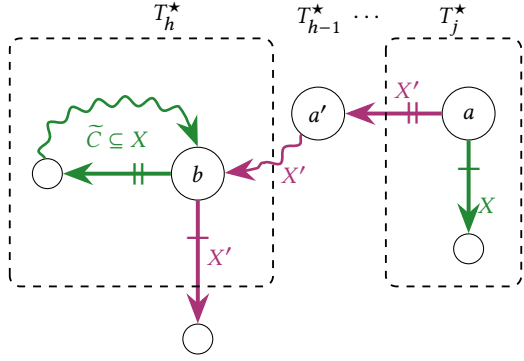
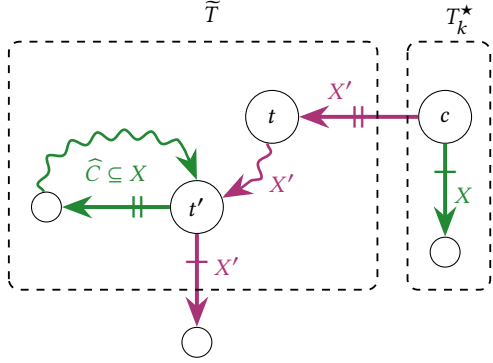

\begin{claim} \label{claim:strategy-proofness}
    For each agent~$a \in \allNN$, we have $X \succeq_a X'$.
\end{claim}
\begin{claimproof}
    For a contradiction, assume that $X' \succ_a X$ for some $a \in \allNN$, and let us choose $a \in \NN_j$ so as to minimize the index $j \in [k-1]$. 
    Since $(a,X(a))$ is undominated in~$H_{|N_j}$, we know that $X'(a)=a'$ for some $a' \in \NN_1 \cup \dots \cup \NN_{j-1}$. 
    Since $X'$ is an allocation, it contains a cycle through the arc~$(a,X'(a))$ entering the set~${\NN_1 \cup \dots \cup \NN_{j-1}}$; 
    hence, it also contains an arc leaving the set~${\NN_1 \cup \dots \cup \NN_{j-1}}$. 
    See Figure~\ref{fig:claim:strategy-proofness} for an illustration.
    Let $(b,X'(b))$ be such an arc, and let $h \in [j-1]$ be the index for which $b \in \NN_h$. Then $\{(v,X(v)):v \in \NN_h\}$ contains a cycle~$\widetilde{C}$ through~$b$. 
    Note that $X \succeq_v X'$ for each agent~$v \in V(\widetilde{C}) \subseteq \NN_h$, due to our choice of~$a$ and $h < j$. 
    Moreover, since $(b,X'(b))$ leaves~$\NN_h$, the arc~$(b,X(b)) \in  \widetilde{C} \subseteq X$ dominates the arc~$(b,X'(b))$ due to the definition of~$E_{S,\T}$ on line~\ref{line:def-A_ST}, i.e., $X(b) \succ_b X'(b)$. 
    Therefore, $\widetilde{C}$ is a blocking cycle for~$X'$ in~$H$.
    Recall also that there is no deviating agent in~$\widetilde{T}$, i.e., $C \cap \allNN=\emptyset$, and thus $\succ_v=\succ'_v$ for each agent~$v \in V(\widetilde{C}) \subseteq \NN_h \subseteq \allNN$. This implies that $\widetilde{C}$ is a blocking cycle for~$X'$ in~$H'$ as well, which contradicts the correctness of Algorithm~\ref{alg:SCFA} established in Theorem~\ref{thm:alg-SCFA-correct}.
\end{claimproof}

We will apply again the ideas used in the proof of Claim~\ref{claim:strategy-proofness}; see Figure~\ref{fig:groupSP} for an illustration. Recall that $(c,X'(c))=(c,t)$ is an arc that enters~$\allNN$, so we know that~$X'$ must also contain an arc that leaves $\allNN$; let $(t',X'(t'))$ be such an arc for some $t' \in \allNN$.
Recall that $(t',X(t'))$ dominates all arcs leaving~$\allNN$, so we have $X \succ_{t'} X'$. By $t' \in \allNN$, we also know that the cycle~$\widehat{C}$ of~$X$ covering~$t'$ only contains agents within~$\allNN$. By Claim~\ref{claim:strategy-proofness}, this implies that $X \succeq_a X'$ for each agent~$a$ on~$\widehat{C}$, and $X \succ_{t'} X'$ for agent~$t'$ on~$\widehat{C}$. This means that $\widehat{C}$ is a blocking cycle for~$X'$ in~$H$.
Recall that since there is no deviating agent in~$\allNN$, i.e., $C \cap \allNN=\emptyset$, we further have $\succ_a=\succ'_a$ for each agent~$a \in V(\widehat{C}) \subseteq \allNN$. 
Hence, $\widehat{C}$ is a blocking cycle for~$X'$ in~$H'$ as well, which contradicts the correctness of Algorithm~\ref{alg:SCFA} established in Theorem~\ref{thm:alg-SCFA-correct}.
\end{proof}

\subsection{Material for Section~\ref{sec:algo-prop}}
\label{app:comparison_QW}
It is known that the Quint--Wako algorithm can return any fixed allocation~$X$ in the strong core of~$H$. In the following lemma, we show that the same holds for Algorithm~\ref{alg:SCFA} if agents' preferences are weak orders. 

\begin{proposition}
\label{prop:alg-on-weakorders}
    Suppose that $H$ is a housing market where agents' preferences are weak orders. 
When run on the instance~$(H,\emptyset)$, 
Algorithm~\ref{alg:SCFA} can find each allocation in the strong core of~$H$, by appropriately choosing between possible allocations on line~\ref{line:Tvalid-def}.
\end{proposition}
\begin{proof}
Suppose that $X$ is an allocation in the strong core of some housing market~$H$. Let $\S$ denote the set of strongly connected components in the subgraph~$D^H[U]$ spanned by the set~$U$ of undominated arcs. For each $S \in \S$, define the arc set $X_{|S}=X \cap (S \times S)$.
The following facts follow from Theorem~\ref{thm:QW-char}:
\begin{itemize}
    \item[(i)] $X_{|S}$ is an allocation in the submarket~$H_{|S}$ for each set~$S \in \S$;
\item[(ii)] if $H_{|S}$ admits an allocation contained in~$U$, then  $X_{|S} \subseteq U$;\item[(iii)] if $S$ is an absorbing set in~$D^H[U]$, then 
$X_{|S} \subseteq U$. 
\end{itemize}

Next, we are going to show that Algorithm~\ref{alg:SCFA} will reach line~\ref{line:return-without-recursion} at some point where the set family~$\T$ contains all absorbing sets in~$D^H[U]$ and possibly some other sets from~$\S$.

To see this, first notice that each absorbing set of~$D^H[U]$ is added to the family~$\T$ on line~\ref{line:add-S-to-T}, due to fact~(iii) above.
Second, by definition, no arc leaving an absorbing set~$S$ of~$D^H[U]$ can be undominated. 
Since agents' preferences are weak orders, all arcs leaving a given agent~$a \in S$ are dominated by all undominated arcs leaving~$a$. Hence, $E_S=E_{S,\T}$ for every possible set family~$\T$ that contains~$S$. This means that no absorbing set is ever added to~$\R$ on line~\ref{line:add-S-to-R}.
Hence, at some point the algorithm reaches line~\ref{line:return-without-recursion}, and when this happens, the set family~$\T$ contains all absorbing sets of~$D^H[U]$, and possibly some other sets from~$\S$ as well. 

Consider the value of the set family~$\T$ at this point, and the set $T^\star=\bigcup_{T \in \T}T$  of agents.
Notice that $T^\star$ must be \emph{closed} in the sense that there cannot be an undominated arc leaving~$T^\star$: indeed, the existence of such an arc~$(a,b)$ would imply that no arc leaving~$a$ would be contained in~$E_{S,\T}$ for the set $S \in \T$ containing~$a$, and hence, $S$ would have been added to~$\R$  on line~\ref{line:add-S-to-R}, resulting in the deletion of~$S$ from~$\T$. 

Using again that agents' preferences are weak orders, the observation that  $T^\star$ is closed, i.e., no undominated arc leaves~$T^\star$, 
implies that all arcs leaving~$T^\star$ are dominated by all arcs in~$U$. 
Therefore, $E_{S,\T} = E_{S} \subseteq U$ for each $S \in \T$.
Moreover, due to line~\ref{line:add-S-to-T}, 
 the submarket~$H_{|S}$ admits an allocation consisting only of undominated arcs. By facts~(i) and~(ii) above, this means that $X_{|S}$ is also an allocation consisting only of undominated arcs. Since there are no forbidden arcs, this means that $X_{|S} \subseteq E_{S,\T}$. Hence, there exists an appropriate choice for the algorithm to store the allocation~$X_{|S}$ on line~\ref{line:Tvalid-def} for each $S \in \T$; assume henceforth that this happens.

In the case $N=T^\star$, this implies that the allocation $\bigcup_{S \in \T} X_{|S}$ returned on line~\ref{line:return-without-recursion} is exactly the allocation~$X$. If $N'=N \setminus T^\star \neq \emptyset$, then $X \cap (N' \times N')$ is clearly in the strong core of~$H_{|N'}$. 
Hence, by an inductive argument we may assume that the recursive call on line~\ref{line:recurse} returns the allocation $X'=X \cap (N' \times N')$ for the submarket~$H_{|N'}$, and thus, the algorithm outputs the allocation 
$\bigcup_{S \in \T} X_{|S} \cup X'=X$ on line~\ref{line:outputX}.
\end{proof}

Contrasting Proposition~\ref{prop:alg-on-weakorders}, we now show that if agents' preferences are partial orders, then there may exist allocations in the strong core that can never be returned by Algorithm~\ref{alg:SCFA}.
Such an instance is presented in Example~\ref{ex4}.

\begin{example}
    \label{ex4}
Consider the following market~$H^6$ over agent set $N=\{a,b,c,d_1,d_2\}$; see its underlying graph without loops on Figure~\ref{fig:ex4}. Let the acceptability sets be defined as $A(a)=A(b)=A(c)=N \setminus \{d_2\}$ and $A(d_1)=A(d_2)=\{d_1,d_2\}$. 
For each agent~${x \in N}$ and $y \in A(x) \setminus \{x\}$, we let $y \succ_x x$;
additionally, we set $d_1 \succ_a b$, $d_1 \succ_b c$ and $d_1 \succ_c a$. 

Consider how Algorithm~\ref{alg:SCFA} runs with $H^6$ as its input without forbidden arcs.  
The set of undominated arcs in the market is $U=\{(b,a),(a,c),(c,a),(d_1,d_2)\} \cup \{(x,d_1):x \in N \setminus \{d_1\}\}$. 
The set of strongly connected components in the subgraph spanned by all undominated arcs is $\S=\{S_1,S_2\}$ where $S_1=\{a,b,c\}$ and  $S_2=\{d_1,d_2\}$. 
 Algorithm~\ref{alg:SCFA} computes on line~\ref{line:def-A_S} the arc sets~$E_{S_1}=\{(a,c),(c,b),(b,a)\}$ and $E_{S_2}=\{(d_1,d_2),(d_2,d_1)\}$, and finds that both~$H^6_{|S_1}$ and~$H^6_{|S_2}$ admit an allocation in~$E_{S_1}$ and in~$E_{S_2}$, respectively. Thus, 
Algorithm~\ref{alg:SCFA} initializes the family~$\T$ by setting $\T=\{S_1,S_2\}$.

Starting the iteration on lines~\ref{line:iter-begin}--\ref{line:iter-end}, the algorithm first computes the arc sets~$E_{S,\T}$ for both $S \in \T$, obtaining  ${E_{S_1,\T}=E_{S_1}}$ and ${E_{S_2,\T}=E_{S_2}}$ and storing the allocations~$X_{S_1}=\{(a,c),(c,b),(b,a)\}$ and~$X_{S_2}=\{(d_1,d_2),(d_2,d_1)\}$ that are the unique allocations for~$H^6_{|S_1}$ in~$E_{S_1,\T}$ and for~$H^6_{|S_2}$ in~$E_{S_2,\T}$, respectively; note that no other allocations can be stored at this point. The algorithm finds on line~\ref{line:return-without-recursion} that $N=T^\star=S_1 \cup S_2$, and hence returns the allocation $X_{S_1} \cup X_{S_2}$. 

Even though Algorithm~\ref{alg:SCFA} can only output the single allocation described above, the strong core of~$H^6$ contains another allocation: namely, the allocation~$\{(a,b),(b,c),(c,a)\} \cup X_{S_2}$. 
\end{example}

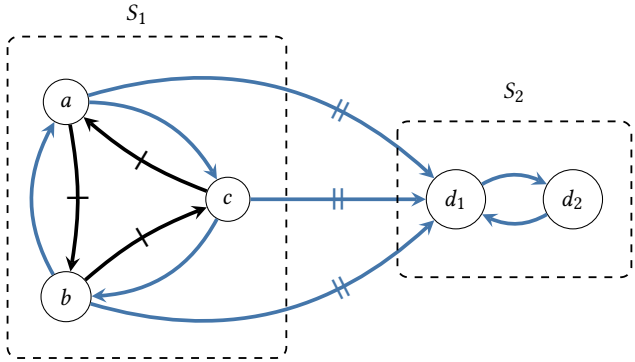
\begin{figure}[ht]
\centering
\resizebox{\columnwidth}{!}{
\begin{tikzpicture}[scale=1.3]
	\node[draw,circle,inner sep=4pt] (a) at (0,1) {$a$};
  	\node[draw,circle,inner sep=4pt] (b) at (0,-1) {$b$};
  	\node[draw,circle,inner sep=4pt] (c) at (1.66,0) {$c$};
	\node[draw,circle,inner sep=4pt] (d1) at (4,0) {$d_1$};  	
    \node[draw,circle,inner sep=4pt] (d2) at (5.2,0) {$d_2$};  	
  	\draw[line width=1.4pt,arrows={-stealth},color=undomcolor] (b) to [bend left=10] node[midway,sloped] {$\boldsymbol{|}$} (c);
  	\draw[line width=1.4pt,arrows={-stealth},color=undomcolor] (c) to [bend left=10] node[midway,sloped] {$\boldsymbol{|}$} (a);
  	\draw[line width=1.4pt,arrows={-stealth},color=undomcolor] (a) to [bend left=10] node[midway,sloped] {$\boldsymbol{|}$} (b);
    \draw[line width=1.4pt,arrows={-stealth},color=domcolor] (a) to [bend left=30] (c);
  	\draw[line width=1.4pt,arrows={-stealth},color=domcolor] (b) to [bend left=30] (a);
  	\draw[line width=1.4pt,arrows={-stealth},color=domcolor] (c) to [bend left=30] (b);
  	\draw[line width=1.4pt,arrows={-stealth},color=domcolor] (a) to [bend left=30] node[pos=0.7,sloped] {$\boldsymbol{|}\boldsymbol{|}$} (d1);
  	\draw[line width=1.4pt,arrows={-stealth},color=domcolor] (b) to [bend right=30] node[pos=0.7,sloped] {$\boldsymbol{||}$} (d1);
  	\draw[line width=1.4pt,arrows={-stealth},color=domcolor] (c) to node[pos=0.5,sloped] {$\boldsymbol{|}\boldsymbol{|}$}  (d1);
  	\draw[line width=1.4pt,arrows={-stealth},color=domcolor] (d1) to [bend left=30] (d2);
  	\draw[line width=1.4pt,arrows={-stealth},color=domcolor] (d2) to [bend left=30] (d1);
  	\draw[rectangle,line width=0.7pt,dashed,rounded corners,above=2pt] (-0.6,1.6) rectangle (2.26,-1.7) {};
  	\node[] at (0.73,1.9) {$S_1$};

    \draw[rectangle,line width=0.7pt,dashed,rounded corners,] (3.4,0.8) rectangle (5.8,-0.8) {};
  	\node[] at (4.6,1.1) {$S_2$};
\end{tikzpicture}
}
\Description{The underlying graph of housing market~$H^6$ defined in Example~\ref{ex4}.}
\caption{The underlying graph of housing market~$H^6$ defined in Example~\ref{ex4}.
}
\label{fig:ex4}
\end{figure}
\section{Additional material for Section~\ref{sec:properties}}

\subsection{Proof of Lemma~\ref{lem:strongcore-solutions-incomp}}
\label{app:proof-of-lemSCincomp}
\lemstrongcoresolutionsincomp*

\begin{proof}
Suppose for the sake of contradiction that $Y \succ_a X$ for some agent~$a \in N$. 
We build a walk~$W$ starting at agent~$a$ using an iterative procedure as follows. The procedure distinguishes between two states, \emph{$X$-preferring} and \emph{$Y$-preferring}, starting initially in a $Y$-preferring state and with $W$ containing zero arcs. 
At each step, we take the agent~$b$ at the end of~$W$ and choose the next arc of~$W$ among the arcs leaving~$b$: 
\begin{itemize}
    \item[(a)] if $Y \succ_b X$, we pick~$(b,Y(b))$ and we switch to (or remain in) the $Y$-preferring state;
    \item[(b)] if $X \succ_b Y$, we pick~$(b,X(b))$ and we switch to (or remain in) the $X$-preferring state;
    \item[(c)] if $X \sim_b Y$, and the procedure is currently in~$Z$-preferring mode for some $Z \in \{X,Y\}$, then we pick $(b,Z(b))$.
\end{itemize}
We stop this procedure at the point where $W$ ceases to be a directed path, because the last arc added to~$W$ points to some vertex already contained in~$W$, thus creating a cycle~$C$. 

Let $u$ be the last agent on~$W$ for whom $W \succ_u X$ or $W \succ_u Y$; possibly $u=a$, so $u$ is well-defined. This means that $u$ is the last agent at which step~(a) or~(b) above was applied. Observe that if $u$ is on~$C$, then $C$ blocks~$X$ or~$Y$, because $W(c) \succeq_c X(c)$ and $W(c) \succeq_c Y(c)$ holds for each $c$ on~$C$ by the definition of~$W$, and $u$ prefers~$C$ to~$X$ or to~$Y$. Hence, this contradicts $X,Y \in \SC(H)$.

Thus, we get that $u$ is not on~$C$, as shown in Figure~\ref{fig:walk_W}. 
Then the arcs put into~$W$ after and including~$(u,W(u))$ form the union of the cycle~$C$ and a path~$P$ leading from~$u$ to some agent~$c$ on~$C$, with $V(P) \cap V(C)=\{c\}$.
Suppose that $W(u)=X(u)$ (the case when $W(u)=Y(u)$ is symmetrical), so the procedure entered (or remained in) the $X$-preferring state at~$u$. By the definition of~$u$, after picking the arc~$(u,X(u))$ at~$u$, the procedure always applied step~(c) above, and thus remained in $X$-preferring mode. Consequently, $W(z)=X(z)$ for all agents $z$ that follow $u$ on~$W$. However, this means that $c$ is entered by two arcs of~$X$, one on~$P$ and one on~$C$ (the last arc put into~$W$). 
This contradicts the assumption that $X$ is an allocation, finishing our proof.
\begin{figure}[th!]
    \centering
    \begin{tikzpicture}[myNode/.style={circle, draw, inner sep=1.5pt, minimum size=24pt},
    myEdge/.style={draw,line width = 1pt, -{Stealth[length=3mm]},  decorate,decoration={snake,post length=3mm,amplitude=0.8mm,segment length=2.4mm}},
    myArc/.style = { draw,line width = 1pt, -{Stealth[length=3mm]}}]
        \node (a) at (1,0) [myNode] {$a$};
        \node (u) at (2.8,0) [myNode] {$u$};        
        \node (Wu) at (4.6,0) [myNode] {$W(u)$};
        \node (c) at (6.6,0) [myNode] {$c$}; 
        \node (CC) at (8.4,0) [myNode] {}; 
       \draw[myEdge] (a) to (u);
       \draw[myArc] (u) to %
       (Wu);
       \draw[myEdge] (Wu) to (c);
       \draw[myEdge] (c) to (CC);\draw[myArc,out=120, in=60] (CC) to node[midway,below=3pt] {$C$} (c); 
        \draw[line width = 1pt,decorate,decoration={brace,mirror,raise=14pt,amplitude=6pt}] (u.west) -- (c.east) node[midway,below=20pt] {$P$};
    \end{tikzpicture}
    \Description{Illustration of walk $W$ for Theorem~\ref{lem:strongcore-solutions-incomp}.}
    \caption{Illustration of walk $W$ for Theorem~\ref{lem:strongcore-solutions-incomp}. Paths are shown as zigzag lines.}
    \label{fig:walk_W}
\end{figure}
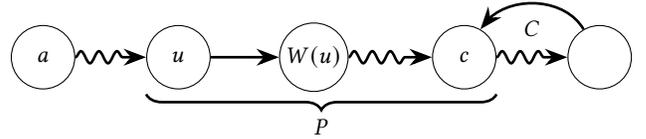
\end{proof}

\subsection{Proof of Theorem~\ref{thm:RI-holds}}
\label{app:proof-of-RI}
\thmRIholds*

\begin{proof}
Let $H'=(N,\{\succ'_a:a \in N\})$ be obtained from~$H$ as the result of a series of $(p,q_i)$-improvements for some agents $q_1,\dots,q_k$; let $Q=\{q_1,\dots,q_k\}$. 

Assume for the sake of contradiction that $X \succ_p X'$.

We will re-use the procedure described in the proof of Lemma~\ref{lem:strongcore-solutions-incomp}, 
applied to allocations~$X$ and~$X'$ with starting point~$p$ and starting in the $X$-preferring state; we use the preferences~$\succ'$ as given for~$H'$.
Let $W$ be the resulting walk. Note that if $W$ does not include an arc pointing to~$p$, then the arguments presented in the proof of Lemma~\ref{lem:strongcore-solutions-incomp} either yield a contradiction to the assumption that $X$ or~$X'$ is an allocation, or prove the presence of a cycle~$C$ that blocks~$X$ or~$X'$ in~$H'$. By $X' \in \SC(H')$, the only possibility is that $C$ blocks~$X$ in~$H'$. Given that $X \in \SC(H)$, we also know that $C$ does not block~$X$ in~$H$. 
By the definition of a $(p,q)$-improvement, we know that this can happen only if $C$ contains some of the arcs~$\{(q_i,p):q_i \in Q\}$, because  the only agents whose preferences differ in~$H$ and~$H'$ are the agents in~$Q$, and moreover, the only agent that can become more preferred by some $q_i \in Q$ when compared to some other agent is agent~$p$. %

Notice now that $(q_i,p) \in C$ for some $q_i \in Q$
implies $C=W$. Recall that $X \succ_p X'$; hence $C(p)=W(p)=X(p)$, and therefore $C$ is a cycle that blocks~$X'$ in~$H'$: all agents on~$C$ weakly prefer~$C$ to~$X'$ by the definition of~$W$, and $p$ strictly prefers~$C$ to~$X$.
This contradiction proves the theorem.
\end{proof}

\subsection{Integer Linear Program for the Strong Core}
\label{app:ILP}
In their paper, Bir\'o et al.~\cite{BKKV-mor} gave an incremental ILP formulation: they first gave an ILP for the core, then added a set of constraints to obtain an ILP for the set of so-called competitive allocations, and finally added another set of constraints for an ILP describing the strong core. 
Since our focus here is solely on the strong core, we show that we can omit certain constraints from this incremental description (namely, those that ensure that the allocation is competitive) and still obtain an ILP that defines the strong core, even when agents' preferences are partial orders. 

Recall our notation for the arc sets $E_{[\succ X]}$ and $E_{[\succeq X]}$ defined in Section~\ref{sec:mychar}.

Let $H=(N,\{\succ_a:a \in N\})$ be a housing market with partial orders with underlying graph $D^H=(N,E)$. Let~$|N|=n$. The ILP defining the strong core of~$H$ is shown in~\ref{ILP:strong-core}. 
For each $(i,j) \in E$, it contains a variable $y_{ij}$ with the following interpretation:
\begin{equation*}
    y_{ij} = \begin{cases}
        1 & \text{if $(i,j)$ is an arc in the allocation,} \\
        0 & \text{otherwise.}
    \end{cases}
\end{equation*}
A further set of variables, $p_i$ for each $i \in N$, 
can be thought of as the \emph{price} of~$i$.
\begin{equation}
\leqnomode
\tag{ILP$_{\texttt{sc}}$} 
\label{ILP:strong-core}
\end{equation}
\begin{align}
    \sum_{j: (i,j)\in E} y_{ij} &=1 && \forall i \in N, \label{eq:ILP_agent} \\
    \sum_{j: (j,i)\in E} y_{ji} &=1 &&  \forall i \in N, \label{eq:ILP_house}\\
   p_i + 1 &\leq p_j + n\cdot \sum_{k: j \preceq_i k} y_{ik} && \forall (i,j)\in E,\label{eq:ILP_price_vector_ineq} \\
     p_i &\leq p_j + n \cdot \sum_{k: j \prec_i k} y_{ik} && \forall (i,j) \in E,\label{eq:ILP_strong_core}\\
   y_{ij} &\in \{ 0, 1\} && \forall i \in N, \label{eq:ILP_binary} \\
    p_i &\in [n] && \forall i \in N.\label{eq:ILP_price_vector_def} 
\end{align}
Constraints (\ref{eq:ILP_agent}) and~(\ref{eq:ILP_house}) together with (\ref{eq:ILP_binary}) determine an allocation~$X$: defining $X(i)=j$ if and only if $y_{ij} = 1$, constraints~(\ref{eq:ILP_agent}) and (\ref{eq:ILP_house}) guarantee that~$X$ contains exactly one arc leaving and entering each agent~$i \in N$, respectively.
Constraints~(\ref{eq:ILP_price_vector_ineq}) and (\ref{eq:ILP_price_vector_def}) ensure that there is no strictly blocking cycle for~$X$, that is, $X$ is in the core.
To see this, recall that $X$ is in the core of~$H$ if and only if arcs of~$E_{[\succ X]}$ form an acyclic subgraph of $D^H$, which in turn happens if and only if there is a topological order of the agents in the subgraph of~$D^H$ spanned by~$E_{[\succ X]}$.
The existence of such a topological order is equivalent to the existence of values~$p_i$, $i \in N$, such that $p_i < p_j$ for each arc $(i,j) \in E_{[\succ X]}$, which is in turn guaranteed by constraints~(\ref{eq:ILP_price_vector_ineq}) and~(\ref{eq:ILP_price_vector_def}).
Finally, adding constraints~(\ref{eq:ILP_strong_core}) narrows down the set of allocations that correspond to the feasible solutions of our ILP to the strong core.

The correctness of~\ref{ILP:strong-core} is expressed by Proposition~\ref{prop:ILP_correct} whose proof relies on Lemma~\ref{lem:core-vs-prices} and 
follows ideas from~\cite{BKKV-mor}. We remark that the ILP formulations of Bir\'o et al.~\cite{BKKV-mor} that describe the core and the set of competitive allocations also remain sound under the generalization that agents' preferences are partial orders.%

To show the correctness of~\ref{ILP:strong-core}, we need the following lemma. 
\begin{lemma}
\label{lem:core-vs-prices}
An allocation $X$ is in the strong core of~$H$ if and only if there exist prices~$p_i \in [n]$ for each $i \in N$ such that
\begin{itemize}
\item $p_i < p_j$ for each $(i,j) \in E_{[\succ X]}$, and 
\item $p_i \leq p_j$ for each $(i,j) \in E_{[\succeq X]}$.
\end{itemize}
\end{lemma}

\begin{proof}
    First, if $X \in \SC(H)$, then there is no cycle consisting of arcs in~$E_{[\succeq X]}$ that also contains an arc in~$E_{[\succ X]}$, as such a cycle would be a blocking cycle.
    Hence, no strongly connected component of~$D_{[\succeq X]}=(N,E_{[\succeq X]})$ may contain a cycle in~$E_{[\succ X]}$, as that would give rise to a blocking cycle for~$X$. 
    
    Let $D'$ be the digraph obtained by contracting each strongly connected component of~$D_{[\succeq X]}$. 
    Let $\varphi(i)$ denote the vertex of~$D'$ corresponding to some $i \in N$.
    On the one hand, we know that $D'$ is acyclic due to its definition. 
    On the other hand, we know that for each arc~$(i,j) \in E_{[\succ X]}$ we have $\varphi(i) \neq \varphi(j)$, and so $(\varphi(i),\varphi(j))$ is an arc in~$D'$. 
    Notice also that $D'$ may further contain arcs that correspond to arcs  in~$E_{[\succeq X]}$ that run between different strongly connected components of~$D_{[\succeq X]}$; such arcs may exist because agents' preferences are partial orders.
    
    Since $D'$ is acyclic, it admits a topological ordering, so we can assign a value~$p_v \in [n]$ for each vertex~$v$ of~$D'$ such that $p_v <p_{v'}$ holds for each arc~$(v,v')$ of~$D'$. This gives rise to prices~$p_i$ for $i \in N$ by setting $p_i=p_{\varphi(i)}$. In this way, we have $p_i <p_j$ for each $(i,j) \in E_{[\succ X]}$, because for such arcs, $(\varphi(i),\varphi(j))$ is an arc of~$D'$. 
    We also know that $p_i \leq p_j$ for each $(i,j) \in E_{[\succeq X]}$, because either $\varphi(i)=\varphi(j)$, which implies $p_i=p_j$, or again $(\varphi(i),\varphi(j))$ is an arc of~$D'$, implying $p_i<p_j$.

    For the other direction of the proof, assume that there are prices~$p_i$ for each $i \in N$ fulfilling the conditions of the lemma. We show that $X$ cannot admit a blocking cycle. Assume for the sake of contradiction that $C$ is a cycle that blocks~$X$. Then $C \subseteq E_{[\succeq X]}$ implies that $p_i \leq p_j$ for each $(i,j) \in C$, which can only happen if $p_i=p_j$ for each $(i,j) \in C$, since $C$ is a cycle.
    However, by the definition of a blocking cycle, $C$ contains an arc~$(i',j')$ of~$E_{[\succ X]}$, for which $p_{i'}<p_{j'}$ follows from the condition on the prices as stated by the lemma, yielding a contradiction. Thus, $X$ is indeed in the strong core of~$H$.
\end{proof}

\begin{proposition}
\label{prop:ILP_correct}
There exists a feasible solution for \textup{\ref{ILP:strong-core}} if and only if the housing market~$H$ admits an allocation in the strong core.
    
\end{proposition}

\begin{proof}
    Let $X$ be an allocation in the strong core of~$H$. Set~$y_{ij}$ as~$1$ whenever each~$(i,j) \in X$, and~$0$ otherwise.
    Then (\ref{eq:ILP_agent}), (\ref{eq:ILP_house}), and~(\ref{eq:ILP_binary}) clearly hold.
    Define prices $p_i$ for each $i \in N$ so that they satisfy the conditions of
    Lemma~\ref{lem:core-vs-prices}. Since $ p_i \in [n]$ for each $i \in N$, constraint~(\ref{eq:ILP_price_vector_def}) is satisfied.

    We show that inequality (\ref{eq:ILP_price_vector_ineq}) holds for each $(i,j) \in E$. 
     On the one hand, if $(i,j) \in E_{[\succ X]}$, then this is clear by our assumption on the prices, as we have $p_i+1 \leq p_j$.
     On the other hand, if $(i,j) \notin E_{[\succ X]}$, that is, $X(i) \succeq_i j$, then $X(i)$ is contained in the set~$\{k: k \succeq_i j \}$. 
     Therefore, we get $\sum_{k: k \succeq_i j } y_{ik} \geq 1$. Since both $p_i$ and~$p_j$ are in~$[n]$, this implies 
    \begin{align*}
        p_i + 1 \leq n + 1 \leq p_j + n \leq p_j + n\cdot \sum_{k: k \succeq_i j} y_{ik}
    \end{align*}
    as desired, proving~(\ref{eq:ILP_price_vector_ineq}).

   Next, we show that (\ref{eq:ILP_strong_core}) holds for each $(i,j) \in E$. 
   On the one hand, if $(i,j) \in E_{[\succeq X]}$, then $p_i \leq p_j$ holds by our choice  of the prices, which implies (\ref{eq:ILP_strong_core}). 
   On the other hand, if $(i,j) \notin E_{[\succeq X]}$, that is, $X(i) \succ_i j$, then $X(i)$ is contained in the set~$\{k:k \succ_i j\}$, and therefore we get 
   $\sum_{k: k \succ_i j} y_{ik} \geq 1$, which implies 
   \[p_i \leq n \leq n \cdot \sum_{k: k \succ_i j} y_{ik} \leq p_j+ n \cdot \sum_{k: k \succ_i j} y_{ik}\]
   proving~(\ref{eq:ILP_strong_core}). 

    Suppose now that values~$y_{ij}$ for each $(i,j) \in E$ together with values $p_i$ for each $i \in N$ yield a feasible solution to~\ref{ILP:strong-core}. Constraints 
    (\ref{eq:ILP_agent}), 
    (\ref{eq:ILP_house}), and
    (\ref{eq:ILP_binary}) determine an allocation~$X$ where an arc $(i,j) \in E$ is contained in~$X$ if and only if $y_{ij}=1$.

    We show that the values~$p_i$, $i \in N$, fulfill the conditions of Lemma~\ref{lem:core-vs-prices}. By constraint~(\ref{eq:ILP_price_vector_def}), we know $p_i \in [n]$ for each $i \in N$.
    
    To show the first condition, assume that $(i,j) \in E_{[\succ X]}$. 
    Then $j \succ_i X(i)$, or in other words, $X(i) \not \succeq_i j$. Thus, 
    $X(i)$ is not contained in $\{k: k \succeq_i j\}$.
    Therefore, $\sum_{k:k \succeq_i j} y_{ik}=0$. Hence by (\ref{eq:ILP_price_vector_ineq}) we know \[p_i+1 \leq p_j + n \cdot \sum_{k:k \succeq_i j} y_{ik}=p_j,\]
    that is, $p_i<p_j$. This proves that the first condition of Lemma~\ref{lem:core-vs-prices} is satisfied.

    Second, assume that $(i,j) \in E_{[\succeq X]}$. 
    Then $j \succeq_i X(i)$, or in other words, $X(i) \not \succ_i j$. Thus, 
    $X(i)$ is not contained in $\{k: k \succ_i j\}$.
    Therefore $\sum_{k:k \succ_i j} y_{ik}=0$. Hence by (\ref{eq:ILP_strong_core}) we know \[p_i \leq p_j + n \cdot \sum_{k:k \succ_i j} y_{ik}=p_j,\]
    that is, $p_i \leq p_j$. This proves that the second condition of Lemma~\ref{lem:core-vs-prices} is satisfied.
    Hence, by Lemma~\ref{lem:core-vs-prices}, $X$ is indeed in the strong core of~$H$.
\end{proof}

\end{appendices}

\fi

\end{document}